\documentclass[twocolumn,amsmath,amssymb,nofootinbib,superscriptaddress]{revtex4-2}
\usepackage[utf8]{inputenc}
\usepackage[margin=1in]{geometry}
\usepackage{amsmath, amsthm, amssymb,amscd, mathrsfs, amsfonts, mathtools,tikz-cd,pgfplots}
\usepackage{graphicx}
\usepackage{subfigure}
\usepackage{subcaption}
\usepackage{subcaption}
\usepackage{relsize}
\usepackage{mathtools}
\usepackage{dsfont}
\usepackage{quantikz}
\usetikzlibrary{shapes, backgrounds}

\graphicspath{ {./images/} } %path relative to main .tex file

\tikzset{
  qaoaBlock/.style={
    draw,
    rounded corners=2pt,
    minimum width=1.4cm,
    minimum height=0.8cm,
    align=center,
    font=\small,
  },
  problem/.style={qaoaBlock, fill=green!30},
  mixer/.style={qaoaBlock, fill=blue!30},
}
\usepackage{ORCIDinREVTeX}
\usetikzlibrary{decorations.pathmorphing}
\usetikzlibrary{arrows.meta}
\usetikzlibrary{patterns}

\usepackage[bookmarks=true,
bookmarksnumbered=true,
breaklinks=true,
pdfstartview=FitH,
hyperfigures=false,
plainpages=false,
naturalnames=true,
colorlinks=true,
linkcolor=blue,       % Color for internal links
citecolor=blue,       % Color for citations
urlcolor=blue,        % Color for URLs
pagebackref=true,
pdfpagelabels]{hyperref}
\theoremstyle{definition}
\newtheorem{thm}{Theorem}[section]
\newtheorem{prop}[thm]{Proposition}
\newtheorem{lem}[thm]{Lemma}

\newtheorem{cor}[thm]{Corollary}
\newtheorem{rmk}[thm]{Remark}

\begin{document}

\title[Provable avoidance of barren plateaus for GM-QAOA]
{Provable avoidance of barren plateaus for the Quantum Approximate Optimization Algorithm with Grover mixers}
%[Dynamical Lie algebras and barren plateaus for GM-QAOA]
%{Dynamical Lie algebras and avoidance of barren plateaus for the Quantum Approximate Optimization Algorithm with Grover Mixers}

\author{Boris Tsvelikhovskiy} 
\orcid{0000-0003-0798-7218}
\email{borist@ucr.edu}
\affiliation{Department of Mathematics, University of California, Riverside, CA 92521, USA}

\author{Matthew Nuyten} 
\affiliation{Department of Mathematics, North Carolina State University, Raleigh, NC 27695, USA}

\author{Bojko N. Bakalov}
\orcid{0000-0003-4630-6120}
%\email{bnbakalo@ncsu.edu}
\affiliation{Department of Mathematics, North Carolina State University, Raleigh, NC 27695, USA}

\date{\today}

%: A Complete Classification

\begin{abstract}
We analyze the dynamical Lie algebras (DLAs) associated with the Grover-mixer variant of the Quantum Approximate Optimization Algorithm (GM-QAOA). 
When the initial state is the uniform superposition of computational basis states, we show that the corresponding DLA is isomorphic to
\(\mathfrak{su}(d) \oplus \mathfrak{u}(1)\oplus \mathfrak{u}(1)\), 
where \(d\) denotes the number of distinct values of the objective function. 
We also establish an analogous result for other choices of initial states and Grover-type mixers.
Furthermore, we prove that the DLA of GM-QAOA has the largest possible commutant among all QAOA variants initialized with the same state, corresponding physically to the maximal set of conserved quantities. 
We derive an explicit formula for the variance of the GM-QAOA loss function in terms of the objective function values, and we show that for a broad class of optimization problems, GM-QAOA with sufficiently many layers avoids barren plateaus.

\end{abstract}

\maketitle

\section{Introduction}\label{sec:I}

Variational quantum algorithms have emerged as a promising approach for tackling complex optimization problems on near-term quantum devices (see \cite{CABB}). At the heart of these methods lies the Quantum Approximate Optimization Algorithm (QAOA), which provides a systematic framework for addressing combinatorial optimization problems \cite{QAOA}. The core idea behind QAOA involves translating classical optimization problems into the quantum domain by encoding the objective function into a linear operator, the problem Hamiltonian $H_P$, so that the ground state (lowest energy state) of $H_P$ corresponds to the optimal solution of the original problem. QAOA also utilizes a mixer Hamiltonian $H_M$, which is needed for facilitating transitions between quantum states. These transitions are crucial for exploring the solution space and escaping local minima. The interplay between $H_P$ and $H_M$, governed by variational parameters, shapes the quantum evolution and guides the system towards the optimal solution.

Despite its theoretical promise, QAOA faces significant practical challenges that limit its scalability and effectiveness. One of the most critical issues is selecting optimal parameters, %$\beta$ and $\gamma$, 
a process that uses classical optimization across a potentially vast landscape. Parameter optimization becomes increasingly difficult as the circuit depth $p$ increases, leading to computational overhead that can offset the quantum advantage \cite{ZWCHL}. A particularly concerning phenomenon is the emergence of barren plateaus, where the variance of the loss function becomes exponentially small with system size, causing gradients to almost vanish and rendering training quantum circuits ineffective \cite{mcclean2018barren,larocca2025barren}.

To address these limitations, researchers have developed more sophisticated approaches (see \cite{BKZS} for a comprehensive overview), most notably the Quantum Alternating Operator Ansatz, which is a generalization of QAOA to accommodate broader classes of mixer Hamiltonians \cite{HWORVB}. These mixer variants offer greater flexibility in navigating the solution space and potentially mitigating barren plateaus. Although the most prominent choice for the mixer Hamiltonian is the $X$-mixer, formulated as the sum of Pauli $X$ operators, this mixer often fails to exploit problem-specific structure. Alternative constructions, including the $XY$-mixer and the Grover mixer have shown promise in recent works (see \cite{BE, BM, HSCHLSP, WRDR, VDAKAL, ZZSGNHLP}). 

Of particular interest in this paper is the family of \emph{Grover mixers} \cite{BE}. 
The original Grover search algorithm, introduced in~\cite{Grover}, provides a quadratic speedup over classical exhaustive search for the unstructured search problem. 
Grover's algorithm proceeds by iteratively applying a reflection about the uniform superposition state followed by a reflection about the marked subspace, thereby amplifying the probability of measuring the target solution. 
A quantum adiabatic version of Grover's algorithm, where the initial Hamiltonian is taken to be the  projector onto the uniform superposition state, was presented in~\cite{RC, DMV} and shown to achieve the same quadratic speedup (see also Theorem~1 in~\cite{FGGN}). 
QAOA algorithms for Grover search were introduced in \cite{MTB,jiang2017near}.

The projector onto the uniform superposition state can serve as a mixer in the QAOA framework. This construction is commonly referred to as Grover-mixer QAOA (GM-QAOA) \cite{BE}. More generally, we use the term GM-QAOA to denote any variant of QAOA in which the mixer is the  projector onto a chosen initial state. This flexibility makes GM-QAOA applicable to \emph{constrained} optimization problems by taking the initial state as a uniform superposition of all feasible states
(see \cite{BE}).
 
Several works have examined the performance and limitations of GM-QAOA (see \cite{APMB,BBGGLE,GBOE1,GBOE2,XXCLSA,ZZPSQKDHP}). Numerical studies on MAX-2-SAT and MAX-3-SAT observed a nonzero reachability deficit, which means that the loss function cannot achieve the optimum unless the depth exceeds a certain threshold %$p^{\ast} = \mathcal{O}(\sqrt{N})$
\cite{APMB,GBOE2}. On the theoretical side, rigorous lower bounds on the depth required to achieve a target approximation ratio were established in~\cite{BBGGLE}, while complementary bounds in terms of the optimality density were provided in~\cite{XXCLSA}. An alternative variant, threshold QAOA, was proposed in~\cite{GBOE1}, where the problem Hamiltonian is replaced by a threshold function; numerical evidence suggests that it can outperform GM-QAOA on some problems. %such as MaxCut, Max $k$-VertexCover, Max $k$-DensestSubgraph, and MaxBisection.

One of the recently developed approaches to studying the performance of variational quantum algorithms, including QAOA, is through the \emph{dynamical Lie algebra} (DLA) generated by the Hermitian operators (Hamiltonians) of the algorithm under the commutator operation. A key insight of this perspective is that the dimension of the DLA is linked to structural properties of the optimization landscape, such as the possible occurrence of barren plateaus~\cite{FHCKYHSP,LJGCC,LCSMCC,RBSKMLC}. Beyond QAOA, DLAs also play an important role in quantum machine learning (QML). In supervised QML, the goal is to identify a circuit within a parametrized family, generated by a finite set of Hamiltonians, that approximates a target function and generalizes to unseen data. The associated DLAs provide rigorous criteria for when barren plateaus occur~\cite{GLCCS,LTWS,WHSU}. Moreover, DLAs offer a powerful framework for analyzing questions of classical simulability~\cite{CLG,GLCCS}.

While previous theoretical analyses of QAOA in general, and of GM-QAOA in particular, have provided valuable insights into performance guarantees and algorithmic behavior, the properties of the DLAs associated with QAOA have remained less well understood. Two recent papers \cite{ASYZ1,KLFCCZ} have studied the DLAs associated with QAOA with the standard $X$-mixer for MaxCut, but only for particular graphs such as paths, cycles, and complete graphs.

The central aim of this paper is to address this gap by carrying out an explicit analysis of the DLAs associated with GM-QAOA. By characterizing these Lie algebras completely, we identify and analyze structural features that directly influence the algorithm's performance.

\subsection{Main Results} 
Our main result concerns the dynamical Lie algebras  arising from GM-QAOAs (see Theorem~\ref{DLAMainThm}).
%In the case where the initial state is the uniform superposition of the computational basis states, we 
We prove that the DLA is isomorphic to 
\(
  \mathfrak{su}(d) \oplus \mathfrak{u}(1)\oplus \mathfrak{u}(1),
\)
% or
% \(
%   \mathfrak{su}(d) \oplus \mathfrak{u}(1),
% \)
where \(d\) denotes the number of distinct objective function values on the initial state, provided that \(d<2^n\). When $d=2^n$, the DLA is the full 
\(
  \mathfrak{u}(2^n).
\)
% A similar description applies to other choices of initial state and Grover-type mixer 
% (see Theorem~\ref{DLAMainThm} for details).

 This structural insight allows us to decompose the Hilbert space into irreducible components under the action of the DLA and precisely characterize the algorithm's expressive power. In addition, we show that the DLAs generated by GM-QAOA possess the largest possible commutant among all DLAs arising from QAOAs with the same initial state, %$\ket{\xi}$, 
 which physically corresponds to having the maximal set of conserved quantities.  
 At the same time, the GM-QAOA DLA generates the smallest associative algebra under inclusion.
 From here, we conclude that other QAOA variants are more expressive than GM-QAOA when acting on the same initial state.

Furthermore, for a broad class of optimization problems, we derive an explicit inverse-polynomial lower bound on the variance of the loss function for sufficiently deep GM-QAOA circuits (Theorem~\ref{thm:BPtheorem}). This result relates the variance to the structural properties and values of the objective function, thereby establishing that GM-QAOA avoids barren plateaus.

We demonstrate these findings on several fundamental optimization problems, including ordinary and weighted MaxCut, Boolean satisfiability (SAT), Max-$k$-VertexCover, and the Traveling Salesperson Problem (TSP).
We also performed numerical experiments for MaxCut, which showed that the variance of the loss function increases and stabilizes with the depth of the GM-QAOA circuit. Combined with our theoretical result, which applies to deep circuits, this suggests that GM-QAOA avoids barren plateaus at any depth.

By establishing conditions under which GM-QAOA avoids barren plateaus, this work provides practical guidance for quantum algorithm designers and opens new avenues for near-term quantum advantage. Our results suggest that careful mixer design can fundamentally alter the training landscape of variational quantum algorithms, potentially enabling quantum optimization approaches to scale to industrially relevant problem sizes. Furthermore, the Lie algebraic framework developed here offers a methodology for analyzing other QAOA variants, potentially accelerating the development of new variational quantum algorithms with favorable training properties.

%\begin{center}
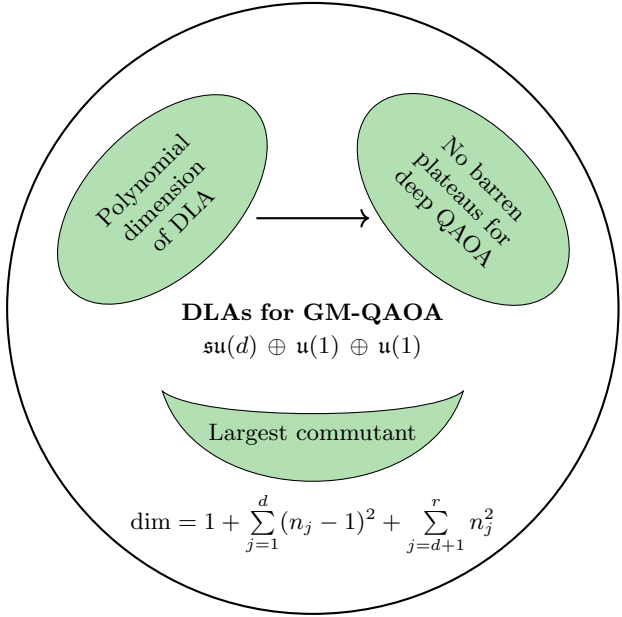
\begin{figure}[h!]
%\begin{center}
\begin{tikzpicture}

\node[align=center, text width=5cm] (above) at (0,0) 
  {\textbf{DLAs for GM-QAOA}\\[2pt]\(\mathfrak{su}(d) \oplus \mathfrak{u}(1) \oplus \mathfrak{u}(1)\)};

  \path[fill=green!60!black!30, draw] 
    (-2,-0.8) .. controls (-1.5,-1.2) and (1.5,-1.2) .. (2,-0.8)   % top curve
    .. controls (1.5,-2.4) and (-1.5,-2.4) .. (-2,-0.8) % bottom curve
    -- cycle;

  % Text in the center
  \node at (0,-1.37) {Largest commutant};

\node at (0,-2.57) {$\dim= 1+\sum\limits_{j=1}^{d} (n_j - 1)^2+\sum\limits_{j=d+1}^{r} n_j^2$};
\node[align=center, text width=2cm, ellipse, draw, fill=green!60!black!30, inner sep=6pt, rotate=-45] 
(below) at (2,1.75) {No barren plateaus for deep QAOA};

\node[align=center, text width=2cm, ellipse, draw, fill=green!60!black!30, inner sep=6pt, rotate=45] 
(below) at (-2,1.75) {Polynomial dimension of DLA};

\draw[thick, ->] (-0.75,1.5) -- (0.75,1.5);

\draw[thick] (0,0.3) circle (1.6in);

\end{tikzpicture}
%\end{center}
\caption{\textbf{Summary of the main results of the paper.} 
We determine the dynamical Lie algebras associated with Grover-mixer QAOA circuits. 
%When the initial state is the uniform superposition of computational basis states, the 
The DLA is isomorphic to 
$\mathfrak{su}(d) \oplus \mathfrak{u}(1) \oplus \mathfrak{u}(1)$, 
where $d<2^n$ is the number of distinct objective function values on the initial state $\ket\xi$.
%An analogous classification holds for arbitrary initial states, with the precise definition of $d$ given in Section~\ref{sec:III}.  
For most combinatorial optimization problems, $d$ grows polynomially in the number of qubits $n$, implying that the DLA dimension is polynomial in $n$. 
As a consequence, the variance of the loss function admits a lower bound of order $1/\mathrm{poly}(n)$, showing that GM-QAOA avoids barren plateaus for high enough circuit depth. Finally, we prove that among all QAOA variants initialized with the same initial state, the GM-QAOA DLA admits the largest commutant, and establish an explicit formula for its dimension in terms of the multiplicities $n_1,\dots,n_r$ of the objective function values.}
\label{fig:GMQAOAavoidsBP}
%\end{center}
\end{figure}
%\end{center}

\subsection{Structure of the Paper}

The paper is organized as follows. Section~\ref{sec:II} gives a brief overview of 
%variational quantum algorithms, with emphasis on 
the Quantum Approximate Optimization Algorithm (QAOA), 
%In particular, Section~\ref{sec:III} 
and introduces, in particular, the Grover-mixer variant of QAOA (GM-QAOA) in the form used throughout this work.

The main results appear in Sections~\ref{sec:IV} and \ref{sec:V}, with proofs deferred to Appendix~\ref{sec:VIII}.
\begin{itemize}
    \item Theorem~\ref{DLAMainThm} gives a complete characterization of the dynamical Lie algebra (DLA) $\mathfrak{g}_{\xi}$ arising from GM-QAOA with a given initial state $\ket{\xi}$.
    \item Theorem~\ref{IsotypicDecomp} establishes the decomposition of the Hilbert space $W=(\mathbb{C}^2)^{\otimes n}$ into a direct sum of isotypic components under the DLA action.
    \item Theorem~\ref{ContainmentThm} shows that, for a given initial state $\ket{\xi}$, the associative algebra generated by the GM-QAOA DLA $\mathfrak{g}_{\xi}$ is the smallest (with respect to inclusion) among associative algebras generated by DLAs with other mixers whose negated Hamiltonians \( -H_M \) have \( \ket{\xi} \) as their ground state.
    As a consequence, $\mathfrak{g}_{\xi}$ has the largest commutant.
    \item Theorem~\ref{LargestCommutantThm} 
    describes the commutant of the GM-QAOA DLA $\mathfrak{g}_{\xi}$, which is defined as the set of all operators commuting with the elements of~$\mathfrak{g}_{\xi}$.
    %shows that, for a given initial state $\ket{\xi}$, the GM-QAOA DLA $\mathfrak{g}_{\xi}$ has the largest possible commutant (with respect to inclusion) among DLAs with other mixers whose negated Hamiltonians \( -H_M \) have \( \ket{\xi} \) as their ground state.
    \item Theorem~\ref{thm:ExpVar} presents explicit formulas for the expectation and variance of the loss function of GM-QAOA, taken with respect to the Haar measure on the dynamical Lie group $e^{\mathfrak{g}_{\xi}}$.
    \item Theorem~\ref{thm:BPtheorem} provides an explicit inverse-polynomial lower bound (in the number of qubits $n$) on the variance of the loss function for a broad class of optimization problems. As a consequence, barren plateaus do not arise when the circuit depth~$p$ is sufficiently large.
\end{itemize}

In Section~\ref{sec:VI}, we apply the above theoretical results to concrete optimization problems. %with a primary emphasis on MaxCut (Sections~\ref{sec:V.numerics} and~\ref{sec:VI.C}).
\begin{itemize}
    \item In Section~\ref{sec:V.search}, we discuss the QAOA algorithms for Grover search \cite{MTB,jiang2017near}, both for the Grover mixer and the $X$-mixer.
    \item In Section~\ref{sec:V.maxcut}, we apply our results to the MaxCut and weighted MaxCut problems, deriving explicit lower bounds on the variance of the loss function. The exact expressions for these bounds are given in Eqs.~\eqref{MaxCutVarBound} and~\eqref{MaxCutVarBound2}.

    \item In Section~\ref{sec:V.numerics}, we perform state vector-based simulations
to evaluate how the expectation and variance of the loss function scale with depth and qubit number, providing evidence for our barren plateau analysis on GM-QAOA.
\item Section~\ref{sec:VI.B} provides a comparative analysis of GM-QAOA and its standard $X$-mixer counterpart in cases where the DLA of the latter is known
(cf.\ \cite{ASYZ1,KLFCCZ}). 
\item Section~\ref{sec:VI.C} considers other optimization problems, to which our results are applicable, such as 
Boolean satisfiability (SAT), Max-$k$-VertexCover, and the Traveling Salesperson Problem (TSP).
%graph colorings and $m$-SAT.
\end{itemize}
Finally, Section~\ref{sec:VII} offers concluding remarks and outlines possible directions for future research.

\section{A Review of the Quantum Approximate Optimization Algorithm}
%{Variational Quantum Algorithms for Binary Optimization Problems}
\label{sec:II}
We provide an overview of %variational quantum algorithms with an emphasis on 
QAOA in a language suitable for the purposes of this paper. For a more comprehensive treatment of the subject,  we refer the reader to \cite{HWORVB, QAOA, BBCC}.

\subsection{From Classical to Quantum}

Let $\mathbb{B}^n := \{0,1\}^n$ denote the set of all binary strings of length $n$.
%, and let $\mathcal{S} := S_{2^n}$ be the symmetric group acting on these $2^n$ elements. 
A broad class of discrete optimization problems can be formulated as the task of finding elements
$x \in \mathbb{B}^n$ that minimize (or maximize) a given objective function
\begin{equation}\label{eq:optfunc}
F\colon \mathbb{B}^n \longrightarrow \mathbb{R},
\end{equation}
possibly subject to additional constraints.

To translate this problem into the quantum setting, consider the complex vector space of $n$ qubits,
$W = (\mathbb{C}^2)^{\otimes n}\cong \mathbb{C}^{2^n}$. The computational basis for $W$ is given by $\{\ket{x} \mid x \in \mathbb{B}^n\}$. 
%indexed by binary strings. 
We define a Hermitian operator $H_P$ on $W$ by
%is said to \emph{represent} the function $F : \mathbb{B}^n \to \mathbb{R}$ if
\begin{equation}\label{eq:HP}
H_P \ket{x} = F(x)\ket{x}, \qquad \forall \, x \in \mathbb{B}^n.
\end{equation}
In other words, $H_P$ is the diagonal operator whose eigenvalues encode the values of $F$.
Then the minimal value of $F$ corresponds to the lowest energy of the \emph{problem Hamiltonian} $H_P$.

% This correspondence leads to the following classical–quantum dictionary:
% \begin{itemize}
%     \item the discrete set $\mathbb{B}^n$ $\rightsquigarrow$ the Hilbert space $W$;
%     \item the objective function $F$ $\rightsquigarrow$ the Hamiltonian $H_P$ representing $F$;
%     \item the absolute minima of $F$ on $\mathbb{B}^n$ $\rightsquigarrow$ the ground states of $H_P$ in $W$.
% \end{itemize}

\subsection{An Overview of QAOA} 
Variational Quantum Algorithms (VQAs) are a prominent class of hybrid quantum-classical algorithms designed to tackle problems in quantum chemistry, machine learning, and combinatorial optimization (see \cite{CABB} for a review).
In a VQA, a parametrized quantum circuit is optimized by a classical algorithm to minimize (or maximize) a loss function, often representing the expectation value of the problem Hamiltonian $H_P$.

One of the most influential instances of VQA is the Quantum Approximate Optimization Algorithm (QAOA), first introduced in \cite{QAOA}.  
%Within the QAOA framework, the Hamiltonian $H_P$ associated with the objective function in (\ref{eq:optfunc}) is commonly referred to as the \emph{problem Hamiltonian}.  
A central component of QAOA is the so-called \emph{mixer Hamiltonian}, denoted $H_M$. The negative of this operator, $-H_M$, is required to have a unique ground state $\ket{\xi} \in W$. %and must meet the criteria outlined in the Perron-Frobenius theorem (see Theorem \ref{PF}). 
The fundamental idea behind QAOA is to gradually deform $H_M$ into $H_P$ through a series of quantum transformations, in such a way that the ground state at each stage is mapped to a ground state of the next.

The process begins with the preparation of the ground state $\ket{\xi}$ of the Hamiltonian $-H_M$, by applying a unitary operator $U_\xi$ to the state $\ket{0\cdots0} = \ket{0}^{\otimes n}$:
\begin{equation}\label{eq:Uxi}
\ket{\xi} = U_\xi \ket{0\cdots0}.
\end{equation}
Subsequently, one applies an alternating sequence of unitary operators generated by the problem Hamiltonian $H_P$ and the mixer Hamiltonian $H_M$
(see Figure~\ref{QAOAcircuit}).
Each layer consists of the action of $U_P(\gamma_j):=e^{-i \gamma_j H_P}$ followed by $U_M(\beta_j):=e^{-i \beta_j H_M}$, where %the real parameters 
\[
\boldsymbol{\beta}=(\beta_1,\dots,\beta_p), \qquad \boldsymbol{\gamma}=(\gamma_1,\dots,\gamma_p),
\] 
with $\beta_j \in [0,\pi)$ and $\gamma_j \in [0,2\pi)$, are collections of classically optimized parameters that control the evolution times under $H_M$ and $H_P$, respectively. Here, $p$ denotes the QAOA \emph{depth}, i.e., the number of alternating layers.

%\begin{center}

\begin{figure}
\resizebox{0.5\textwidth}{!}{\begin{tikzpicture}
% blue frame (everything except measurement)
\node[draw, rounded corners, fill=blue!30, minimum width=7.3cm, minimum height=4cm, inner sep=0.5cm, anchor=center, label=below:{Ansatz circuit}] (blueframe) at (-0.3,0) {};

% red frame for measurement part
\node[draw, rounded corners, fill=red!30, minimum width=1cm, minimum height=4cm, inner sep=0.5cm, anchor=center] (redframe) at (4.2,0) {};

\node (circ) {
\begin{quantikz} 
\lstick{\ket{0}} & \gategroup[4,steps=1,style={rounded corners,fill=black!20,draw}, label style={yshift=-0.75in}]{\scriptsize$U_\xi$} & & \gategroup[4,steps=2, style={rounded corners,fill=green!20,draw}, label style={yshift=-.75in}]{\scriptsize$U_P(\gamma_1)$} & & & \gategroup[4,steps=2,style={rounded corners,fill=yellow!20,draw}, label style={yshift=-0.75in}]{\scriptsize$U_M(\beta_1)$} & & \cdots&  \gategroup[4,steps=2, style={rounded corners,fill=green!20,draw}, label style={yshift=-0.75in} ]{\scriptsize$U_P(\gamma_p)$} & & & \gategroup[4,steps=2,style={rounded corners,fill=yellow!20,draw}, label style={yshift=-0.75in}]{\scriptsize$ U_M(\beta_p)$} &  & &\meter{} \\ 
 \lstick{\ket{0}} & & & & & & & & \cdots & & & & & & &\meter{}\\
 \lstick{\vdots} \\
 \lstick{\ket{0}} & & & & & & & &  \cdots & & & & & & & \meter{} \end{quantikz}
};
% arrow 
\draw[<-]  (-1.9,-2.2) -- (-1.9,-3);
\draw[-]  (-1.9,-3) -- (4.1,-3) node[midway,below,yshift=-0.1cm,black]{Update of parameters $(\boldsymbol{\beta}, \boldsymbol{\gamma})$};
\draw[-]  (4.1,-2.2) -- (4.1,-3);
\end{tikzpicture}}
\caption{\textbf{Schematic illustration of the QAOA circuit.} The initial state preparation unitary $U_\xi$ is followed by $p$ alternating applications of unitaries $U_P(\gamma_j):=e^{-i \gamma_j H_P}$ and $U_M(\beta_j):=e^{-i \beta_j H_M}$, generated by
the problem Hamiltonian $H_P$ and the mixer Hamiltonian $H_M$, respectively. At the end of the circuit, the state is measured in the computational basis. 
Each measurement outcome $x \in \mathbb{B}^n$ is assigned the value $F(x)$ of the objective function, 
and the empirical mean of these values provides an estimate of 
$\bra{\psi(\boldsymbol{\beta},\boldsymbol{\gamma})} H_P \ket{\psi(\boldsymbol{\beta},\boldsymbol{\gamma})}$. 
This estimate is then used in a classical optimization loop to update the parameters $(\boldsymbol{\beta}, \boldsymbol{\gamma})$ with the goal of minimizing the empirical mean.}
\label{QAOAcircuit}
\end{figure}
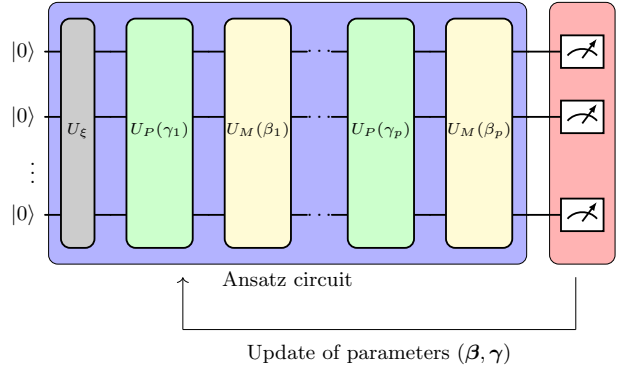
%\end{center}

The overall transformation operator is
\begin{equation}
\begin{aligned}
U&(\boldsymbol{\beta},\boldsymbol{\gamma})
:= U_M(\beta_p) U_P(\gamma_p) \cdots U_M(\beta_1) U_P(\gamma_1) \\
&= e^{-i \beta_p H_M} e^{-i \gamma_p H_P} \cdots e^{-i \beta_1 H_M} e^{-i \gamma_1 H_P}.
\label{qaoa-chain}
\end{aligned}
\end{equation}
At the end of the protocol, the circuit prepares the state
\begin{equation}\label{eq:psibg}
   \ket{\psi(\boldsymbol{\beta},\boldsymbol{\gamma})} 
      := U(\boldsymbol{\beta},\boldsymbol{\gamma}) \ket{\xi}.
\end{equation}
This state is measured in the computational basis. 
Each measurement produces a bit string $x \in \mathbb{B}^n$ with probability 
$\Pr(x) = |\braket{x}{\psi(\boldsymbol{\beta},\boldsymbol{\gamma})}|^2$, 
and is assigned the value $F(x)$ of the original objective function. By repeating the procedure for many runs and averaging the observed values, one obtains an empirical estimate of the expectation value of the problem Hamiltonian:
\begin{equation}\label{eq:empmean}
   %\mathbb{E}[F(x)] 
   \langle F \rangle
      = \bra{\psi(\boldsymbol{\beta},\boldsymbol{\gamma})} H_P \ket{\psi(\boldsymbol{\beta},\boldsymbol{\gamma})}.
\end{equation}
This empirical mean serves as the \emph{loss function} that guides the classical optimization of the parameters $(\boldsymbol{\beta},\boldsymbol{\gamma})$.

%A schematic overview of the algorithm is given in Figure~\ref{QAOAcircuit}.

Although the problem Hamiltonian $H_P$, which encodes the objective function via Eq.~\eqref{eq:HP}, is uniquely determined by the classical problem, the freedom to choose the Hamiltonian mixer $H_M$ remains. %Under the assumption that $H_M$ satisfies the Perron-Frobenius conditions, the convergence of QAOA to a classical solution minimizing $F$ is ensured by Theorem \ref{PF} below. \begin{defn} An $n\times n$ matrix $M$ with real coefficients is called \textbf{irreducible} iff there are no proper $M$-invariant coordinate subspaces of $\mathbb{C}^n$: there is no subspace $\{0\}\neq W\subset \mathbb{C}^n$ such that $M(W)\subseteq W$.\end{defn} \begin{thm} (Perron-Frobenius). Let $M=(m_{ij})\in \mbox{Mat}_n(\mathbb{R})$ be an irreducible matrix with $m_{ij}\geq 0$.\begin{itemize}\item Then there is a positive real number $r$, such that $r$ is an eigenvalue of $M$ and any other eigenvalue $\lambda$ (possibly complex) has $\operatorname{Re}(\lambda)<r$. \item Moreover, there exists a unique vector $v=(v_1,v_2,\dots,v_n)$ such that $Mv=rv$ and $v_1+v_2+\dots+v_n=1$. This vector is positive, i.e., all $v_i$ are positive real numbers. \end{itemize}\label{PF}\end{thm}   The properties of a Hamiltonian mixer sufficient for convergence of the corresponding QAOA are exactly those appearing in Theorem \ref{PF}  (see \cite{BKZS}). 
The most prevalent choice for the mixer Hamiltonian is the sum of Pauli $X$ operators: 
\begin{equation}\label{Xmixer}
B=\sum\limits_{j=1}^n X_j,
\end{equation}
where $n$ denotes the number of qubits used to encode the problem. Although this choice is generic and convenient, it does not exploit any structure specific to the problem at hand.

Several alternative constructions for the mixer Hamiltonian have been proposed in the literature. For example, in \cite{HWORVB}, the authors introduced the framework of the Quantum Alternating Operator Ansatz, which accommodates broader classes of mixers. Their approach is particularly useful for optimization problems that involve both hard constraints, defining a \emph{feasible subspace} of $W$, and soft constraints that should be optimized.
Furthermore, numerical simulations in \cite{GPSK} have shown that the use of mixers formed by linear combinations of Pauli $X$ and $Y$ operators can improve performance in low-depth QAOA circuits. Further examples and generalizations can be found in \cite{BFL,GBOE1,WRDR} and the references therein.

\subsection{QAOA with Grover mixer}
In the present work, we focus on the variant of QAOA that employs the 
\emph{Grover mixer}, defined as the  projector onto the initial state:
\begin{equation}\label{eq:GM}
    G_M := \ket{\xi}\bra{\xi} = U_\xi \ket{0\cdots0}\bra{0\cdots0} U_\xi^\dagger
\end{equation}
(cf.\ Eq.~\eqref{eq:Uxi}).
This variant, introduced in~\cite{BE}, is commonly referred to as the 
\emph{Grover-mixer QAOA} (GM-QAOA).  
For unconstrained optimization problems, a natural choice is to take $\ket\xi$
to be the uniform superposition of computational basis states:
\begin{equation}\label{eq:xi++}
%\ket{\xi} = 
\ket{+\cdots+} = \frac{1}{\sqrt{2^n}}\sum_{x\in\mathbb{B}^n}\ket{x}.
\end{equation}
For \emph{constrained} optimization problems, one can take
$\ket{\xi}$ to be an equal superposition of a basis of the feasible subspace.

%By definition of the unitary generated by the Grover mixer $G_M$ subject to a parameter $\beta\in [0,\pi)$, is 
The unitary operator $U_M(\beta) = e^{-i\beta G_M} =e^{-i\beta \ket\xi \bra\xi}$ can be implemented as 
\begin{equation}\label{eq:GM_unitary}
    U_M(\beta) 
    = U_\xi \bigl(I - (1-e^{-i\beta})\ket{0\cdots 0}\bra{0\cdots 0}\bigr) U_\xi^\dagger.
\end{equation} 
%The development and implementation details for this variant are found in the work of~\cite{BE}.
For simulation purposes, we construct $U_M(\beta)$ via a series of gate applications to a state vector simulator in PennyLane~\cite{pennylane}. 

Note that, for unconstrained optimization problems such as MaxCut, the unitary that prepares the initial state $\ket\xi = \ket {+\cdots +}$ is simply 
$U_\xi = H^{\otimes n}$, while a more careful treatment is required for constrained optimization problems. 
Crucially, $U_\xi$ must be constructible with a polynomial-depth circuit, otherwise the approach is intractable (cf.\ \cite{BE}).
For instance, the Max-$k$-Vertex Cover problem requires Dicke states $\ket\xi = \ket{D_k^n}$, which represent an equal superposition of all possible computational basis vectors of Hamming weight $k$ in an $n$ qubit system.  Dicke states can be prepared in $\mathcal{O}(nk)$ gates and $\mathcal{O}(n)$ depth~\cite{BE2}; thus, they are practical for GM-QAOA.

\section{Dynamical Lie Algebras for GM-QAOAs}\label{sec:IV}

We define the \emph{dynamical Lie algebra} (DLA) associated with a given QAOA instance as the real Lie algebra generated by the skew-Hermitian operators $iH_M$ and $iH_P$, where $H_M$ and $H_P$ denote the mixer and problem Hamiltonians, respectively. 
Equivalently, the DLA is the smallest real Lie algebra (with respect to the commutator bracket) that contains both $iH_M$ and $iH_P$. 
As we will see in the following sections, the structure of this Lie algebra provides important insight into the expressive power and inherent limitations of the QAOA ansatz.

\subsection{Hilbert Space Decomposition and Matrices of Restricted Hamiltonians}
%{QAOA with Grover Type Mixers}
\label{sec:III}

Let $\{\lambda_1, \dots, \lambda_{r}\}$ denote the set of distinct values in the range of the objective function $F\colon \mathbb{B}^n \to \mathbb{R}$.
The Hilbert space $W\cong \mathbb{C}^{2^n}$ of $n$ qubits then admits a natural orthogonal decomposition according to the level sets of $F$:
\begin{equation}
   W = W_{\lambda_1} \oplus W_{\lambda_2} \oplus \cdots \oplus W_{\lambda_{r}},
\label{LevelSetDecomp}
\end{equation}
where each subspace $W_{\lambda_j}$ is spanned by the computational basis states that attain the value $\lambda_j$ of the objective function, namely  
\[
\bigl\{ \ket{x} \mid x \in \mathbb{B}^n,\; F(x) = \lambda_j \bigr\}.
\]
Equivalently, $W_{\lambda_j}$ are the eigenspaces of the problem Hamiltonian $H_P$.
We let $n_j := \dim W_{\lambda_j}$. %= |F^{-1}(\lambda_j)|$.

With respect to the decomposition~\eqref{LevelSetDecomp}, the initial state $\ket\xi$ admits an expansion
\begin{equation}
  \ket{\xi} = \sum_{j=1}^{r} c_j \ket{\xi_j},  
\label{InitStateDecomp}
\end{equation}
where $\ket{\xi_j} \in W_{\lambda_j}$ is a normalized vector with $\braket{\xi_j}{\xi_j}=1$ if $c_j \neq 0$, and is defined to be zero otherwise.  
In other words, the component $\ket{\xi_j}$ represents the normalized projection of $\ket\xi$ onto $W_{\lambda_j}$ whenever that projection is nonzero.  

For each $j$ such that $\ket{\xi_j} \neq 0$, we further refine $W_{\lambda_j}$ by splitting off the one-dimensional space spanned by $\ket{\xi_j}$:
\[
\widetilde{W}_{\lambda_j} := W_{\lambda_j} \cap \bigl(\mathbb{C}\ket{\xi_j}\bigr)^{\perp} \subset W_{\lambda_j}.
\]

Relabeling indices, we may assume that $\ket{\xi_j} \neq 0$ for $1 \le j \le d$ (with $d \le r$), and $\ket{\xi_j} = 0$ for $j > d$.  
This convention gives %(for $1 \le j \le d$) 
the orthogonal decompositions
\[
W_{\lambda_j} = \mathbb{C}\ket{\xi_j} \oplus \widetilde{W}_{\lambda_j}, \qquad 1 \le j \le d,
\]
where $\dim \widetilde{W}_{\lambda_j} = n_j - 1$.  
Hence, we can rewrite the decomposition~\eqref{LevelSetDecomp} in a more structured form:
\begin{equation}
   W %= \bigoplus_{j=1}^{d} \Bigl( \widetilde{W}_{\lambda_j} \oplus \mathbb{C}\ket{\xi_j} \Bigr) \oplus\bigoplus_{j=d+1}^{r} W_{\lambda_j} \\
    = \widetilde{W} \oplus W_0,
\label{simplifiedDecompForGmax}
\end{equation}
where
\begin{align}
\widetilde{W} &:= \bigoplus_{j=1}^{d} \widetilde{W}_{\lambda_j} \oplus\bigoplus_{j=d+1}^{r} W_{\lambda_j}, 
\label{eq:Wtilde}
\\
W_0 &:= \bigoplus_{j=1}^{d} \mathbb{C}\ket{\xi_j}.
\label{eq:W0}
\end{align}

To summarize, $W$ splits into two natural components as follows.
\begin{itemize} 
\item 
The subspace $W_0$ has an orthonormal basis \( \{ \ket{\xi_1}, \dots, \ket{\xi_d} \} \), given by the nontrivial components of $\ket\xi$ along the eigenspaces of $H_P$. In particular, \(\ket{\xi} \in W_0\).
\item 
The subspace $\widetilde{W} = (W_0)^\perp$ contains all vectors orthogonal to $\ket{\xi_1}, \dots, \ket{\xi_d}$. Note that $\widetilde{W} \subset \bigl(\mathbb{C}\ket{\xi}\bigr)^{\perp}$.
\end{itemize}

The splitting \eqref{simplifiedDecompForGmax} is particularly well suited for analyzing the action of the Grover mixer, since $G_M= \ket{\xi}\bra{\xi}$ acts as the identity on $\ket\xi$ and annihilates its orthogonal complement.
Hence, $G_M$ acts trivially on $\widetilde{W}$, while its restriction to $W_0$ is represented by the matrix
\begin{equation}\label{RestrictionsToTriv2}
G_{M, 0} :=
\begin{pmatrix}
c_1 \bar{c}_1 & c_1 \bar{c}_2 & \cdots & c_1 \bar{c}_d \\
c_2 \bar{c}_1 & c_2 \bar{c}_2 & \cdots & c_2 \bar{c}_{d} \\
\vdots & \vdots & \ddots & \vdots \\
c_{d} \bar{c}_1 & c_{d} \bar{c}_2 & \cdots & c_d \bar{c}_d
\end{pmatrix},
\end{equation}
relative to the basis \( \{ \ket{\xi_1}, \dots, \ket{\xi_d} \} \).
This follows directly from Eq.~\eqref{InitStateDecomp} and $\braket{\xi_j}{\xi} =c_j$.

On the other hand, the operator \( H_P \) is diagonal in the basis \( \{ \ket{\xi_1}, \dots, \ket{\xi_d} \} \). The corresponding matrix is
\begin{equation}\label{RestrictionsToTriv1}
H_{P, 0} :=
\begin{pmatrix}
\lambda_1 & 0 & \cdots & 0 \\
0 & \lambda_2 & \cdots & 0 \\
\vdots & \vdots & \ddots & \vdots \\
0 & 0 & \cdots & \lambda_{d}
\end{pmatrix}.
\end{equation}

%
%Let \( \ell \) denote the number of indices \( j \) such that \( n_j > 1 \); that is, the number of fibers which give rise to nontrivial irreducible components. Then, by Wedderburn’s theorem, the dimension of the image of the group algebra \( \mathbb{C}[G_{\text{max}}] \subseteq \operatorname{End}(W) \) is given by: \(k^2 + \ell,\)
%
%\begin{rmk}
% Note that the vectors \( \{ \ket{\xi_1}, \dots, \ket{\xi_d} \} \)  form an orthonormal basis for \( W_{0} \), and the initial state \(\ket{\xi}\) is also in \( W_{0} \) (by Eq.~\eqref{InitStateDecomp}).
%\end{rmk}

%\begin{rmk}Let \( H_M \) be any mixer Hamiltonian whose action on \( W \) commutes with the full symmetry group \( G_{\text{max}} \). Then the commutant of the Lie algebra \( \mathfrak{g}_{H_M, H_P} \) contains the image of the group algebra \( \mathbb{C}[G_{\text{max}}] \) in \( U(W) \).  This follows from the assumption that both \( H_M \) and \( H_P \) commute with the action of \( G_{\text{max}} \), implying that the entire group algebra acts by symmetries of the system and thus lies in the commutant of any Lie algebra generated by \( H_M \) and \( H_P \).\end{rmk}

\subsection{Complete Description of Dynamical Lie Algebras Corresponding to GM-QAOA}

We are now prepared to describe the DLA in the case where the mixer Hamiltonian is given by the Grover mixer $G_M= \ket{\xi}\bra{\xi}$. We define the DLA
\[
    \mathfrak{g}_{\xi} := \langle iG_M,\, iH_P \rangle_{\mathrm{Lie}}
\]
to be the real Lie algebra generated by the operators $iG_M$ and $iH_P$.

\begin{thm}
%Let $c_1, \dots, c_{d}$ be the coefficients in the decomposition of the initial state $\ket{\xi}$ from~\eqref{InitStateDecomp}. 
The dynamical Lie algebra of GM-QAOA, \( \mathfrak{g}_{ \xi}:=\langle iG_M,\, iH_P \rangle_{\mathrm{Lie}} \), admits the following description:
\[
\mathfrak{g}_{\xi} \;\cong\;
\begin{cases}
\mathfrak{su}(d) \oplus \mathfrak{u}(1) \oplus \mathfrak{u}(1), & \text{if }\, d < 2^n, \\
   \mathfrak{su}(d) \oplus \mathfrak{u}(1), & \text{if }\, d = 2^n,
\end{cases}
\]
where $d$ is the number of nonzero summands in the expression \eqref{InitStateDecomp} of $\ket\xi$ in terms of eigenvectors of $H_P$.
\label{DLAMainThm}
\end{thm}

\begin{rmk}
   In the case of GM--QAOA with initial state $\ket\xi = \ket{+\cdots+}$, Theorem~1 of \cite{BLSAMS} established an upper bound of $d^{2}+1$ for the dimension of $\mathfrak{g}_{ \xi}$. 
   The description obtained in the preceding theorem is consistent with this bound. 
\end{rmk}

Recall that the \emph{center} of the dynamical Lie algebra, 
\(
\mathfrak{z} \subset \mathfrak{g}_{ \xi},
\) 
is defined as the set of all elements that commute with every element of $\mathfrak{g}_{ \xi}$. Equivalently, it consists of all elements that commute with the generators $G_M$ and $H_P$. The next result explicitly describes the center of the DLA $\mathfrak{g}_{\xi}$. 

\begin{cor}
Under the assumptions of Theorem~\ref{DLAMainThm}, for $d < 2^n$, the center $\mathfrak{z}$ of the Lie algebra $\mathfrak{g}_{\xi}$ is $2$-dimensional. It is spanned over $\mathbb R$ by $iH_P \pi_{\widetilde{W}}$ and $i\pi_{W_0}$, where
\begin{equation}\label{eq:proj}
\pi_{\widetilde{W}} \colon W \to \widetilde{W}, \qquad 
\pi_{W_0} \colon W \to W_0
\end{equation}
denote the orthogonal projections corresponding to the decomposition \eqref{simplifiedDecompForGmax}.
%the restriction of $H_P$ to $\widetilde{W}$ and by the identity operator on $W_0$, both extended as linear operators on $W$.
\label{cor:center}
\end{cor}

% In the above corollary, we use the following convention.
% The orthogonal decomposition \eqref{simplifiedDecompForGmax} allows us to decompose any linear operator $T\colon W\to W$ as a sum of four linear transformations acting between the subspaces $\widetilde{W}$ and $W_0$, corresponding to writing a matrix in block form. In particular, we can view an operator $T_0\colon W_0\to W_0$ as an operator on $W$ by extending it to act trivially on $\widetilde{W}$. Note that the identity operator on $W_0$, when extended trivially to $\widetilde{W}$, becomes the orthogonal projection $P_{W_0}$ from $W$ onto $W_0$.
% Similarly, any $\widetilde{T} \colon \widetilde{W} \to \widetilde{W}$ can be extended to act trivially on $W_0$.

\subsection{Decomposition of the Hilbert Space under the Action of the DLA}

Another fundamental question concerns the decomposition of the Hilbert space  \( W \) under the action of the GM-QAOA dynamical Lie algebra \( \mathfrak{g}_{ \xi} \). 
%This 
%%decomposition can be described explicitly using the structural characterization of \( \mathfrak{g}_{ \xi} \) established earlier and 
%can be derived as a consequence of the proof of Theorem~\ref{DLAMainThm}, which is given in Appendix \ref{sec:VIII}.

\begin{thm}
    As a representation of the Lie algebra \( \mathfrak{g}_{ \xi} \), the Hilbert space \( W = \mathbb{C}^{2^n} \) decomposes as
    \[
        W = W_{0} \oplus \mathbb{C}^{\oplus (2^n - d)},
    \]
    where \( W_{0} \) is an irreducible \( d \)-dimensional representation and the remaining summands form a direct sum of \( 1 \)-dimensional representations.
    %\( \mathfrak{g}_{ \xi} \)-invariant subspaces.
    \label{IsotypicDecomp}
\end{thm}

Now consider another DLA $\mathfrak{g} = \langle iH_M, iH_P \rangle_{\mathrm{Lie}}$ generated by the same problem Hamiltonian $H_P$ and a different mixer $H_M$. We want to compare its expressiveness to that of \( \mathfrak{g}_{ \xi} \), i.e., to understand how the decomposition of $W$ under the action of $\mathfrak{g}$ compares to the decomposition from Theorem \ref{IsotypicDecomp}. We start with the observation that a representation of a Lie algebra on $W$ is equivalent to a representation of the associative algebra generated by it.

Let us denote by $\operatorname{End}(W)$ the set of all $\mathbb C$-linear operators on $W$.
Note that $\operatorname{End}(W)$ is a vector space over $\mathbb C$ and is closed under the product given by composition, so it is an associative algebra.
For any subset $S\subset\operatorname{End}(W)$, we define $\mathcal{A}(S)$ to be the \emph{associative closure} of $S$; this is the smallest subalgebra of $\operatorname{End}(W)$ containing $S$, and can be obtained as the linear span of all products of elements of $S$. We also say that $S$ generates $\mathcal{A}(S)$ as an associative algebra.

Similarly, the \emph{Lie closure} $\langle S\rangle_{\mathrm{Lie}}$ is defined as the smallest real Lie algebra containing $S$. It is clear that $\langle S\rangle_{\mathrm{Lie}} \subset \mathcal{A}(S)$ and
$\mathcal{A}(\langle S\rangle_{\mathrm{Lie}}) = \mathcal{A}(S)$.

As a consequence from the proof of Theorem~\ref{DLAMainThm}, presented in Appendix \ref{sec:proofIII.1},
we obtain a description of the associative closure of $\mathfrak{g}_{\xi}$.

\begin{cor}\label{cor:Agxi}
We have 
\[\mathcal{A}(\mathfrak{g}_{\xi}) = \operatorname{End}(W_0)\pi_{W_0} \oplus \mathbb{C}[H_P \pi_{\widetilde{W}}],
\]
where $\pi_{W_0}$, $\pi_{\widetilde{W}}$ are the orthogonal projections \eqref{eq:proj}.
\end{cor}

Our next result shows that the associative closure of any DLA associated with a QAOA circuit for a given problem Hamiltonian and initial state contains the associative closure of the corresponding Grover-mixer QAOA.

\begin{thm}\label{ContainmentThm}
Let \( H_M \) be any mixer Hamiltonian with \( \ket{\xi}\) the unique ground state of \(-H_M\), and
$\mathfrak{g} = \langle iH_M, iH_P \rangle_{\mathrm{Lie}}$ be the corresponding DLA.
Then there exists a polynomial $g(t) \in\mathbb{R}[t]$ such that $G_M = g(H_M)$.
As a consequence, 
\begin{equation}\label{eq:AgAgxi}
\mathcal{A}(\mathfrak{g}_\xi) \subseteq \mathcal{A}(\mathfrak{g}).
\end{equation}
%the associative closure $\mathcal{A}(\mathfrak{g})$ contains the associative closure $\mathcal{A}(\mathfrak{g}_\xi)$
%of the GM-QAOA DLA \( \mathfrak{g}_{\xi} \). 
% Let $\ket{\xi}=\ket{+\dots+}$ and \(\mathfrak{g}_{\mathrm{std}}:=\langle iB,\; iH_P\rangle_{\mathrm{Lie}}\) be the DLA for QAOA with the standard mixer \(B=\sum\limits_{j=1}^n X_j\). 
% As a consequence, the irreducible \(\mathfrak{g}_\xi\)-representation \(W_0\) %from Theorem~\ref{IsotypicDecomp} 
% is contained in an irreducible \(\mathfrak{g}\)-representation \(V_0\).
\end{thm}

The above theorem applies to the case where $H_M=B$ is the standard $X$-mixer $B$ from Eq.\ \eqref{Xmixer}, and the \emph{standard DLA} \(\mathfrak{g}_{\mathrm{std}}:=\langle iB,\; iH_P\rangle_{\mathrm{Lie}}\), for the initial state $\ket\xi=\ket{+\cdots+}$. Therefore,
\begin{equation}\label{eq:Agst}
\mathcal{A}(\mathfrak{g}_{+\cdots+}) \subseteq \mathcal{A}(\mathfrak{g}_{\mathrm{std}}).
\end{equation}
Note, however, that the inclusion \eqref{eq:Agst} does not imply an inclusion of the corresponding Lie algebras;
see Section \ref{sec:VI.B} below for examples in which we compare them.

From the containment \eqref{eq:AgAgxi}, we derive the following:
\begin{cor}\label{cor:containment}
%Under the assumptions of Theorem \ref{ContainmentThm}, 
Let \( H_M \) be any mixer Hamiltonian with \( \ket{\xi}\) the unique ground state of \(-H_M\), and
$\mathfrak{g} = \langle iH_M, iH_P \rangle_{\mathrm{Lie}}$ be the corresponding DLA.
Then the irreducible \(\mathfrak{g}_\xi\)-representation \(W_0\) %from Theorem~\ref{IsotypicDecomp} 
is contained in an irreducible \(\mathfrak{g}\)-representation \(V_0\).
\end{cor}

This corollary has a number of useful consequences.  
First, note that the initial state \(\ket{\xi}\) belongs to \(W_0\), and hence also to \(V_0\).  
Since \(W_0\) (respectively, \(V_0\)) is preserved by the action of the DLA \(\mathfrak{g}_\xi\) (respectively, \(\mathfrak{g}\)),
the same is true for the corresponding dynamical Lie groups. Hence,
the state of the circuit (prior to measurement) remains inside \(W_0\) for GM--QAOA and inside \(V_0\) for the other QAOA
(see Eq.~\eqref{qaoa-chain} and Figure~\ref{QAOAcircuit}).

%The decomposition also has a concrete implication for the location of solution states.  
Denote by \(\lambda_{\min}\) the minimal value of the objective function $F$ as in~\eqref{eq:optfunc}, and let
\[
    \ket{\xi_{\lambda_{\min}}}
    := \frac{1}{\sqrt{n_{\min}}}\sum_{\substack{x\in\mathbb{B}^n\\ F(x)=\lambda_{\min}}}\ket{x}
\]
be the normalized uniform superposition over all optimal solutions, where \(n_{\min}\) is the number of such bit strings.
A priori, it is not clear whether \(V_0\) contains any eigenvectors of \(H_P\) with eigenvalue \(\lambda_{\min}\). 
However, as \(\ket{\xi_{\lambda_{\min}}}\in W_0\) and \(W_0 \subseteq V_0\), it follows that \(\ket{\xi_{\lambda_{\min}}} \in V_0\). 

As the QAOA dynamics preserve \(V_0\) and the initial state is in \(V_0\), it is possible to approximate the state \(\ket{\xi_{\lambda_{\min}}}\) with the state $\ket{\psi(\boldsymbol{\beta},\boldsymbol{\gamma})} = U(\boldsymbol{\beta},\boldsymbol{\gamma}) \ket{\xi}$
with any desired precision, for a suitable choice of parameters $(\boldsymbol{\beta},\boldsymbol{\gamma})$ and sufficiently large circuit depth $p$.
Consequently, a computational-basis measurement performed on \(\ket{\xi_{\lambda_{\min}}}\) will give a solution \(x \in \mathbb{B}^n\) with \(F(x) = \lambda_{\min}\) with a high probability.

\subsection{Commutant of the Dynamical Lie Algebra}

The \emph{commutant} of a subset $S\subset\operatorname{End}(W)$ is defined as the set $\mathcal{C}(S)$ of all $T\in\operatorname{End}(W)$ that commute with all elements of $S$. Note that $\mathcal{C}(S)$ is an associative algebra. 
It is clear from the definitions that
\begin{equation}\label{eq:CS}
\mathcal{C}(S) = \mathcal{C}(\langle S\rangle_{\mathrm{Lie}}) = \mathcal{C}(\mathcal{A}(S)).
\end{equation}

% The operations of taking associative closure, Lie closure, and commutant satisfy the following properties.
% \begin{prop}\label{prop:AC}
% For any subset $S\subset\operatorname{End}(W)$, we have:
% \begin{enumerate}
%     \item $\mathcal{C}(S) = \mathcal{C}(\langle S\rangle_{\mathrm{Lie}}) = \mathcal{C}(\mathcal{A}(S))$;
%     \item $\mathcal{A}(S) = \mathcal{C}(\mathcal{C}(S))$.
% \end{enumerate}
% \end{prop}
% The first property in the above proposition follows easily from the definitions. The second property is a finite-dimensional version of von Neumann's Double Commutant Theorem (see e.g.\ \cite{procesi2007lie}, Theorem 6.2.5).

Now consider a DLA $\mathfrak{g} = \langle iH_M, iH_P \rangle_{\mathrm{Lie}}$ generated by Hamiltonians $H_M$, $H_P$. Then its commutant $\mathcal{C}(\mathfrak{g}) = \mathcal{C}(H_M,H_P)$ is the set of all linear operators commuting with $H_M$ and $H_P$. In the physical setting, elements of the commutant correspond to \emph{conserved quantities}: operators that remain invariant under the dynamics generated by the DLA. 

It follows from Theorem \ref{ContainmentThm} that the commutant of any DLA associated with a QAOA circuit for a given problem Hamiltonian and initial state is contained in the commutant of the corresponding Grover-mixer QAOA.

\begin{cor}\label{cor:maxcomm}
Let \( H_M \) be any mixer Hamiltonian with \( \ket{\xi}\) the unique ground state of \(-H_M\). Then
$\mathcal{C}(H_M,H_P) \subseteq \mathcal{C}(G_M,H_P)$, where $G_M = \ket\xi\bra\xi$ is the Grover mixer.
\end{cor}

We can explicitly describe the commutant of the GM-QAOA DLA. In particular, the dimension of the commutant quantifies the number of independent symmetries preserved by the system. 

\begin{thm}\label{LargestCommutantThm}
Let \( H_P \) be the problem Hamiltonian, and \( \{ W_{\lambda_j} \}_{1 \leq j \leq r} \) be its eigenspaces, with \( \dim (W_{\lambda_j}) = n_j \). Fix an initial state $\ket\xi$, and define $\widetilde{W}$ and $W_0$ by \eqref{eq:Wtilde}, \eqref{eq:W0}.
%\begin{enumerate}
%    \item 
Then the commutant of the GM-QAOA DLA \( \mathfrak{g}_{\xi} \) is
\begin{equation}\label{eq:Cgxi}
\begin{aligned}
\mathcal{C}(\mathfrak{g}_{\xi}) 
&= \bigoplus_{j=1}^{d} \operatorname{End}(\widetilde{W}_{\lambda_j}) \pi_{\widetilde{W}_{\lambda_j}} \\
&\oplus\bigoplus_{j=d+1}^{r} \operatorname{End}(W_{\lambda_j}) \pi_{W_{\lambda_j}}
\oplus \mathbb{C} \pi_{W_0},
\end{aligned}
\end{equation}
where $\pi$ denote the orthogonal projections corresponding to the decompositions \eqref{simplifiedDecompForGmax},
\eqref{eq:Wtilde}.
%    \item 
We have
\[
\dim \mathcal{C}(\mathfrak{g}_{\xi}) =  \sum_{j=1}^{d} (n_j - 1)^2+\sum_{j=d+1}^{r} n_j^2+1.
\]
    %\item The commutant of the Lie algebra \( \mathfrak{g}_{ \xi} \) attains this maximal dimension.
    % \[
    % \dim\left( \mathrm{Comm}(\mathfrak{g}_{ \xi}) \right) = 1+\sum_{j=1}^{d} (n_j - 1)^2+\sum_{j=d+1}^{r} n_j^2.
    % \]
%\end{enumerate}
\end{thm}

The above theorem simply means that the commutant $\mathcal{C}(\mathfrak{g}_{\xi})$ consists of block diagonal matrices relative to the decompositions \eqref{simplifiedDecompForGmax}, \eqref{eq:Wtilde}.

The proofs of Theorems \ref{DLAMainThm}, \ref{IsotypicDecomp}, \ref{ContainmentThm}, \ref{LargestCommutantThm} are presented in Appendix \ref{sec:VIII}.

\section{Loss Function Statistics and the Barren Plateau Phenomenon}\label{sec:V}

In this section, we investigate the asymptotic behavior of the variance of the loss function for QAOA circuits with a mixer Hamiltonian of type $G_{M}$. 
Analyzing this variance is essential for understanding the emergence of barren plateaus in the optimization landscape \cite{larocca2025barren, RBSKMLC}.

Recall that the \emph{loss function} in QAOA is defined as the expectation value of the problem Hamiltonian $H_P$ with respect to the parametrized quantum state:
\begin{equation}
\begin{aligned}
\ell_{\boldsymbol{\beta}, \boldsymbol{\gamma}} (\rho, H_P) 
    :=&\, \bra{\psi(\boldsymbol{\beta},\boldsymbol{\gamma})} H_P \ket{\psi(\boldsymbol{\beta},\boldsymbol{\gamma})} \\
    =&\, \operatorname{Tr} \left[ U(\boldsymbol{\beta}, \boldsymbol{\gamma}) \, \rho \, U^\dagger(\boldsymbol{\beta}, \boldsymbol{\gamma}) \, H_P \right],
\end{aligned}
    \label{LossFnEqn}
\end{equation}  
where $\rho = \ket{\xi}\bra{\xi}$ denotes the pure-state density operator 
associated with the initial state $\ket{\xi}$
(see Eqs.\ \eqref{qaoa-chain}, \eqref{eq:psibg}, \eqref{eq:empmean}).
%and the parametrized quantum circuit $U(\boldsymbol{\beta},\boldsymbol{\gamma})$ is given in \eqref{qaoa-chain}.
More generally, we can define the loss function
\begin{equation}
    \ell_{\boldsymbol{\beta}, \boldsymbol{\gamma}} (\rho, O) = \operatorname{Tr} \left[ U(\boldsymbol{\beta}, \boldsymbol{\gamma}) \, \rho \, U^\dagger(\boldsymbol{\beta}, \boldsymbol{\gamma}) \, O \right]
    \label{LossFnEqn2}
\end{equation}  
for any observable $O$.

%We will assume that the observable $O$ acts diagonally on the computational basis $\{\ket{x} \mid x \in \mathbb{B}^n\}$, so that it commutes with $H_P$.

Let $\zeta_{\Lambda}$ be a random variable uniformly distributed over the set of values  
\(
\left\{\lambda_1, \lambda_2, \dots, \lambda_d \right\},
\) 
attained by the objective function on the vectors appearing in the decomposition \eqref{InitStateDecomp} of the initial state $\ket\xi$.
The expected value and variance of $\zeta_{\Lambda}$ are given by:
\begin{align}
\mathbb{E}(\zeta_{\Lambda}) &= \frac{1}{d} \sum_{1 \le j \le d} \lambda_j, \label{eq:expectation}\\
\operatorname{Var}(\zeta_{\Lambda}) &= \frac{1}{d^2} \sum_{1 \le i < j \le d} (\lambda_i - \lambda_j)^2. \label{eq:variance}
\end{align}
We can express the variance of the GM-QAOA loss function as follows.

\begin{thm}\label{thm:ExpVar}
For sufficiently large circuit depth $p$, the variance of 
$\ell_{\boldsymbol{\beta}, \boldsymbol{\gamma}}(\rho, H_P)$ 
over the parameters $(\boldsymbol{\beta}, \boldsymbol{\gamma})$ is approximated by
\begin{equation} 
\operatorname{Var}_{\boldsymbol{\beta}, \boldsymbol{\gamma}}\bigl[\ell_{\boldsymbol{\beta}, \boldsymbol{\gamma}} (\rho, H_P) \bigr] 
\approx \frac{\operatorname{Var}(\zeta_{\Lambda})}{d+1}.  
\label{varianceFormula} 
\end{equation}
Similarly, the expected value of the loss function satisfies
\begin{equation}
\mathbb{E}_{\boldsymbol{\beta}, \boldsymbol{\gamma}}\bigl[\ell_{\boldsymbol{\beta}, \boldsymbol{\gamma}} (\rho, H_P) \bigr]
\approx \frac{\mathbb{E}(\zeta_{\Lambda})}{d}.
\label{ExpectationOfLossFunction}
\end{equation}
\end{thm}
The above asymptotic formulas hold when the circuit depth $p$ is sufficiently large. The precise threshold of $p$ for which they become valid is established in Theorems~2 and~3 of~\cite{RBSKMLC}.
\label{LossFuncVariance}

The phenomenon of \emph{barren plateaus} in variational quantum algorithms refers to the exponential suppression of gradients in the parameter landscape, which renders classical training methods ineffective \cite{mcclean2018barren,larocca2025barren}. A key diagnostic of this behavior is the variance of the loss function over the parameter space: when this variance is exponentially small in the number of qubits \( n \), the corresponding gradient magnitudes also vanish exponentially \cite{arrasmith2021equivalence}.

We say that an objective function 
\(
F\colon \mathbb{B}^n \to \mathbb{R}
\)
is \emph{\(s\)-local} if it can be written as
\begin{equation}\label{sLocalFunc}
    F(x) = \sum_{j=1}^{T} f_j\big(x_{S_j}\big),
\end{equation}
where each component function 
\(
f_j \colon \mathbb{B}^{s} \to \mathbb{R}
\)
depends only on the subset of bits indexed by \( S_j \subseteq \{1, \dots, n\} \) with \( |S_j| = s \). 
The total number of such local terms is \( T \le \binom{n}{s} \).

The next theorem establishes an explicit lower bound on the variance of the normalized loss function for GM-QAOAs with underlying $s$-local objective functions. This bound demonstrates that, for sufficiently large circuit depth, the variance does not vanish exponentially in \(n\), and therefore these instances avoid barren plateaus.

In order to introduce the normalization of $H_P$, let
\[
|\lambda_{\max}| := \max_{x \in \mathbb{B}^n} |F(x)|
\]
be the maximum absolute value of the objective function \(F(x)\). We define the normalized problem Hamiltonian
\[
\widehat{H}_P := \frac{1}{|\lambda_{\max}|} H_P
\]
with respect to the spectral (or operator) norm.

\begin{thm}\label{thm:BPtheorem}
Suppose that the objective function \(F(x)\) is \(s\)-local and takes only integer values, %\(\lambda_j \in \mathbb{Z}\), 
and the initial state \(\ket{\xi}\) is not an eigenvector of the problem Hamiltonian \(H_P\). Let 
\[
M := \max_{1 \le j \le T} \max_{y \in \mathbb{B}^s} |f_j(y)|
\]
denote the maximum absolute value attained by the local component functions \(f_j\) from Eq.\ \eqref{sLocalFunc}.
% , and let
% \[
% \lambda_{\max} := \max_{x \in \mathbb{B}^n} F(x)
% \]
% be the maximum value of the objective function \(F(x)\). Define the normalized problem Hamiltonian
% \[
% \widehat{H}_P := \frac{H_P}{\lambda_{\max}}
% \]
% with respect to the spectral norm.
Then, for sufficiently large circuit depth \(p\), the variance of the loss function 
\(\ell_{\boldsymbol{\beta}, \boldsymbol{\gamma}}(\rho, \widehat{H}_P)\) over the parameters \((\boldsymbol{\beta}, \boldsymbol{\gamma})\) satisfies the lower bound
\begin{equation}
\label{VarLowerBound} 
\operatorname{Var}_{\boldsymbol{\beta}, \boldsymbol{\gamma}}\bigl[\ell_{\boldsymbol{\beta}, \boldsymbol{\gamma}} (\rho, \widehat{H}_P) \bigr]
\geq 
\frac{(s!)^2}{12\, M^2\, n^{2s}}.
\end{equation}

\end{thm}

The proofs of Theorems \ref{thm:ExpVar} and  \ref{thm:BPtheorem} can be found in Appendix \ref{sec:VIII}.

%\begin{rmk}It has been shown that, for the MaxCut problem on archetypal graphs—namely, graphs that are connected but neither bipartite nor cycles—the variance of the QAOA loss function with the standard mixer (i.e., the transverse-field \( X \)-mixer) decays exponentially with \( n \) (see Corollary~2 in~\cite{KLFCCZ}). In contrast, Theorem~\ref{LossFuncVariance} indicates that employing Grover-type mixers can yield a substantial improvement in the trainability of QAOA by ensuring that the variance decays only polynomially.\end{rmk}

\section{Applications}\label{sec:VI}
One of the most natural choices for the initial state in QAOA is the uniform superposition of all computational basis states,
\(
\ket{\xi} = \ket{+\cdots+}
\)
(see Eq. \eqref{eq:xi++}).
In this case, the Grover mixer \( G_{M} \)  acts by \( 1 \) on \( \ket{\xi} \) and by \( 0 \) on the hyperplane orthogonal to \( \ket{\xi} \) with respect to the standard Hermitian inner product (i.e., the subspace of states whose coordinate sum is zero).

\begin{rmk}
The Grover mixer \( G_{M} \) admits the Pauli-\( X \) representation
\begin{equation}
    \begin{aligned}
    G_{M} &= \frac{1}{2^n}(I + X_1)(I + X_2) \cdots (I + X_n)\\
     &=\frac{1}{2^n}\sum_{l=0}^n e_l(X_1, \dots, X_n),
    \end{aligned} \label{eq:grover}
\end{equation}
where \( e_l(X_1, \dots, X_n) \) is the degree-\( l \) elementary symmetric polynomial in \( X_1, \dots, X_n \) given by
\begin{equation}\label{eq:el}
e_l(X_1, \dots, X_n) = \sum_{1 \le i_1 < i_2 < \cdots < i_l \le n} X_{i_1} X_{i_2} \cdots X_{i_l},
\end{equation}
with the convention that \( e_0(X_1, \dots, X_n) = I \) is the identity operator.
\end{rmk}

For each value \( \lambda_j \) of the objective function \( F \), we define the normalized uniform vector over the corresponding level set as
\begin{equation}
    \ket{\xi_j} := \frac{1}{\sqrt{n_j}} \sum_{x \,:\, F(x) = \lambda_j} \ket{x}, 
    \qquad j = 1,  \dots, r.
\label{BasisOfTrivRep}
\end{equation}
The set \( \{ \ket{\xi_1}, \dots, \ket{\xi_{r}} \} \) forms an orthonormal basis for \( W_{0} \).  
The uniform superposition state \( \ket{\xi} \) can be expressed as
\[
\ket{\xi} = \frac{1}{\sqrt{2^n}}\sum_{j=1}^r\sqrt{n_j}\,\ket{\xi_j},
\]
so that in the notation of Eq.\ \eqref{InitStateDecomp} the coefficients are \( c_j = \sqrt{\frac{n_j}{2^n}} \) and \( d = r\). 
Relative to the basis \( \{ \ket{\xi_1}, \dots, \ket{\xi_{r}} \} \),
the restrictions of the cost Hamiltonian \( H_P \) and the Grover mixer \( G_M \) to the subspace \( W_0 \) are (cf.\ \eqref{RestrictionsToTriv2}, \eqref{RestrictionsToTriv1}):
\begin{align}\label{RestrictionsToTriv}
 H_{P,0} &=
\begin{pmatrix}
\lambda_1 & 0 & \cdots & 0 \\
0 & \lambda_2 & \cdots & 0 \\
\vdots & \vdots & \ddots & \vdots \\
0 & 0 & \cdots & \lambda_{r}
\end{pmatrix},\\
G_{M,0} &=
\frac{1}{2^n}
\begin{pmatrix}
n_1 & \sqrt{n_1 n_2} & \cdots & \sqrt{n_1 n_{r}} \\
\sqrt{n_2 n_1} & n_2 & \cdots & \sqrt{n_2 n_{r}} \\
\vdots & \vdots & \ddots & \vdots \\
\sqrt{n_{r} n_1} & \sqrt{n_{r} n_2} & \cdots & n_{r}
\end{pmatrix}.
\notag
\end{align}
Assuming $r<2^n$, the corresponding dynamical Lie algebra $\mathfrak{g}_\xi$ is isomorphic to 
\(
\mathfrak{su}(r) \oplus \mathfrak{u}(1)\oplus \mathfrak{u}(1),
\)
by Theorem~\ref{DLAMainThm}.

\subsection{Grover Search with QAOA}\label{sec:V.search}

Recall that the classical Grover search algorithm \cite{Grover} addresses the task of identifying a marked element among $2^n$ possibilities, where the set $\mathcal M$ of marked elements may contain one or several items. The algorithm achieves query complexity $\mathcal{O}(\sqrt{2^n/|\mathcal{M}|})$, where $|\mathcal{M}|$ denotes the number of marked elements. 

The QAOA formulation of Grover search \cite{MTB} %(see also \cite{jiang2017near}), Grover's algorithm 
employs the mixer Hamiltonian $G_M=\ket\xi\bra\xi$, with the initial state chosen as $\ket{\xi} = \ket{+\cdots+}$, and the problem Hamiltonian $H_P= \ket\omega\bra\omega$ given by the projector onto the uniform superposition of all marked states,
\[
\ket\omega = \frac{1}{\sqrt{|\mathcal{M}|}} \sum_{x \in \mathcal{M}} \ket{x}.
\]

Since $H_P$ has only two eigenvalues, $0$ and $1$, by Theorem \ref{DLAMainThm}, the corresponding dynamical Lie algebra  is
\[
\mathfrak{g}_{ \xi} \cong \mathfrak{su}(2) \oplus \mathfrak{u}(1) \oplus \mathfrak{u}(1).
\]
Moreover, Theorem \ref{IsotypicDecomp} tells us that, as a representation of \( \mathfrak{g}_\xi\), the Hilbert space \( W = (\mathbb{C}^2)^{\otimes n} \) decomposes as a direct sum of an irreducible \(2\)-dimensional representation spanned by $\ket\omega$ and $\ket\xi$, and a sum of \( 1 \)-dimensional representations.

\begin{rmk}
Although the Lie algebra \( \mathfrak{g}_\xi\) is only of dimension $5$, the classical simulation techniques based on $\mathfrak{g}$-sim described in \cite{GLCCS}  are not directly applicable (see also \cite{CLG}). The subtlety lies in the fact that the precise embedding of $\mathfrak{g}_{ \xi}$ into the ambient Lie algebra $\mathfrak{u}(2^n)$ is not explicitly known.
\end{rmk}

Another approach to performing Grover search with QAOA was proposed in \cite{jiang2017near}. Let us describe it in the case where we search for a single marked element $\omega \in \mathbb{B}^n$. The problem Hamiltonian $H_P = \ket\omega\bra\omega$ is again the projector onto the marked state $\ket\omega$, while the mixer $H_M$ is taken to be the standard $X$-mixer $B$ from Eq.\ \eqref{Xmixer}. They generate the standard DLA \(\mathfrak{g}_{\mathrm{std}}:=\langle iB,\; iH_P\rangle_{\mathrm{Lie}}\).

We observe that, after conjugation with $H^{\otimes n}$ that swaps Pauli $X$ and $Z$, we can treat $B$ as the problem Hamiltonian
\[
H_Z := \sum_{j=1}^n Z_j,
\]
while $H_P$ can be viewed as a Grover mixer $G_M=\ket\xi\bra\xi$ with $\ket\xi=H^{\otimes n}\ket\omega$. Therefore, the results of Theorems \ref{DLAMainThm} and \ref{IsotypicDecomp} apply. Moreover, it is easy to see that $H_Z$ has eigenvalues $-n,-n+2,\dots,n-2,n$.
Since $\ket\xi=\ket{s_1s_2\cdots s_n}$ for some $s_i=\pm$, all eigenvalues of $H_Z$ are attained in the expansion of $\ket\xi$ in terms of the computational basis. Hence, $d=n+1$ and we obtain the following:

\begin{cor}\label{cor:search}
The standard dynamical Lie algebra corresponding to QAOA for Grover search of a single marked state $\ket\omega$ is 
\[
\mathfrak{g}_{\mathrm{std}} \cong \mathfrak{su}(n+1) \oplus \mathfrak{u}(1) \oplus \mathfrak{u}(1).
\]
As a representation of \( \mathfrak{g}_{\mathrm{std}}\), the Hilbert space \( W = (\mathbb{C}^2)^{\otimes n} \) decomposes as a direct sum of an irreducible \((n+1)\)-dimensional representation containing $\ket\omega$, and a sum of \( 1 \)-dimensional representations.
\end{cor}

\subsection{Analysis of GM-QAOA for the MaxCut and Weighted MaxCut Problems}\label{sec:V.maxcut}
The \emph{MaxCut} problem is a well-known combinatorial optimization problem defined on an undirected graph \( G = (V, E) \).  
Given a binary assignment of vertices to two disjoint sets, the goal is to maximize the number of edges that have endpoints in different sets.  
Equivalently, the problem seeks a partition of the vertex set \( V = \{1, \dots, n\} \) into two parts such that the number of \emph{crossing edges}---those whose endpoints belong to different subsets---is maximized.  

For a binary string \( x = (x_1, \dots, x_n) \in \mathbb{B}^n \), where \( x_i = 0 \) if and only if vertex \( i \) is placed in the first subset of the cut, the objective function for MaxCut is given by:
\[
F(x) = \sum\limits_{(i,j) \in E} \bigl((1 - x_i)x_j + x_i(1 - x_j)\bigr).
\]
Here, each term in the sum contributes \( 1 \) if the endpoints \( i \) and \( j \) lie on opposite sides of the cut, and \( 0 \) otherwise.

%
% \begin{prop}
%     Let \( G = (V, E) \) be an undirected graph with vertex set \( V = \{1, 2, \dots, n\} \), and let 
%     \[
%     F(x_1, \dots, x_n) = \sum\limits_{(i,j)\in E} \bigl((1 - x_i)x_j + x_i(1 - x_j)\bigr)
%     \]
%     be the objective function for the MaxCut problem, where each \( x_i \in \{0,1\} \) indicates the side of the cut to which vertex \( i \) belongs. Then the range of \( F \) is a subset of \( \{0, 1, 2, \dots, |E|\} \).
%     \label{maxcutRangeProp}
% \end{prop}
%
% \begin{proof}
%     First, observe that for any pair \( (i,j) \in E \), the expression \( (1 - x_i)x_j + x_i(1 - x_j) \) is equal to 1 if \( x_i \neq x_j \), and 0 if \( x_i = x_j \). 
%     So the summand contributes $1$ if and only if the edge \( (i,j) \) is cut, i.e., its endpoints lie in different parts of the partition defined by \( \{x_i\} \). 
%
%     Therefore, \( F(x_1, \dots, x_n) \) simply counts the number of edges crossing the cut defined by the binary string \( (x_1, \dots, x_n) \). Since each edge contributes either $0$ or $1$ to the sum, and there are \( |E| \) total edges, it follows that:
%     \[
%     0 \leq F(x_1, \dots, x_n) \leq |E|,
%     \]
%     implying that the image of \( F \) is a subset of the integers in this range:
%     \[
%     \operatorname{Range}(F) \subseteq \{0, 1, 2, \dots, |E|\}.
%     \]
% \end{proof}

Note that the range of \( F \) is a subset of \( \{0, 1, 2, \dots, |E|\} \). 
Hence,
\[
d = \dim(W_0) \le |E| + 1 \le \binom{n}{2} + 1.
\]

It is instructive to derive a lower bound on the variance of the loss function from Theorem~\ref{thm:BPtheorem}. 
Since the objective function is $2$-local and each constituent term 
\[
(1 - x_i)x_j + x_i(1 - x_j)
\] 
attains a maximum value of $M = 1$, 
Eq.~\eqref{VarLowerBound} yields the bound
\begin{equation}\label{MaxCutVarBound}
    \operatorname{Var}_{\boldsymbol{\beta}, \boldsymbol{\gamma}}\big[\ell_{\boldsymbol{\beta}, \boldsymbol{\gamma}} (\rho, \widehat{H}_P) \big] \ge \frac{2^2}{12 \, n^4} = \frac{1}{3 n^4}.
\end{equation}

An interesting variation of the MaxCut problem is the \emph{weighted MaxCut} problem, which is widely considered in QAOA since edge weights encode application‑specific costs and generalize a larger class of optimization problems. Consequently, recent studies have developed graph‑decomposition methods to fit large weighted instances to NISQ devices~\cite{herrman2023graph}; derived formulas to motivate weight-aware parameter choices~\cite{ZWCHL}; and derived analytical expressions for weighted hypergraph problems that employ the many-body Grover-type mixers~\cite{NKK}.

The formulation of weighted MaxCut is similar to that of MaxCut. Consider an undirected graph $G=(V,E)$ with edge weights $w_{ij}=w_{ji}>0$ corresponding to edges $(i,j)\in E$, and $w_{ij}=0$ for $(i,j)\not\in E$. The objective is to partition the vertex set \( V = \{1, \dots, n\} \) into subsets labeled $0$ and $1$, such that the weighted sum of edges corresponding to different subsets in the partition is maximized, leading to the $2$-local objective function 
\[
F(x) = \sum\limits_{(i,j) \in E} w_{ij} \bigl((1 - x_i)x_j + x_i(1 - x_j)\bigr).
\]

Setting $w_{\mathrm{max}}:=\max w_{ij}$, implies that the constituent terms in $F$ attain a maximum value of $M = w_{\mathrm{max}}$. Thus, Eq.~\eqref{VarLowerBound} yields the bound
\begin{equation}\label{MaxCutVarBound2}
    \operatorname{Var}_{\boldsymbol{\beta}, \boldsymbol{\gamma}}\big[\ell_{\boldsymbol{\beta}, \boldsymbol{\gamma}} (\rho, \widehat{H}_P) \big] \ge \frac{1}{3w_{\mathrm{max}}^2n^4}.
\end{equation}
%%
%Furthermore, Proposition~\ref{maxcutRangeProp} implies that, for the MaxCut problem, the dimension of the subspace \(W_{0}\) satisfies \(\dim W_{0} \le |E| + 1\), and hence grows at most quadratically with \(n\).

\subsection{Numerical Simulations for MaxCut}\label{sec:V.numerics}

We carried out a series of numerical experiments to empirically validate our analytical results. Namely, we analyze the loss function for the MaxCut problem, to investigate the dependence of the variance on circuit depth $p$ and provide an estimate for sufficient depth $p$ in which GM-QAOA avoids barren plateaus (see Theorem~\ref{thm:BPtheorem}). We use a noiseless state vector simulation in PennyLane~\cite{pennylane}, to allow access to exact probability distributions and expectation values.

\begin{figure*}
\centering
\begin{subfigure}{}
%\caption{Expectation vs. Depth.}\label{fig:exp_depth}
\includegraphics[width=0.45\textwidth]{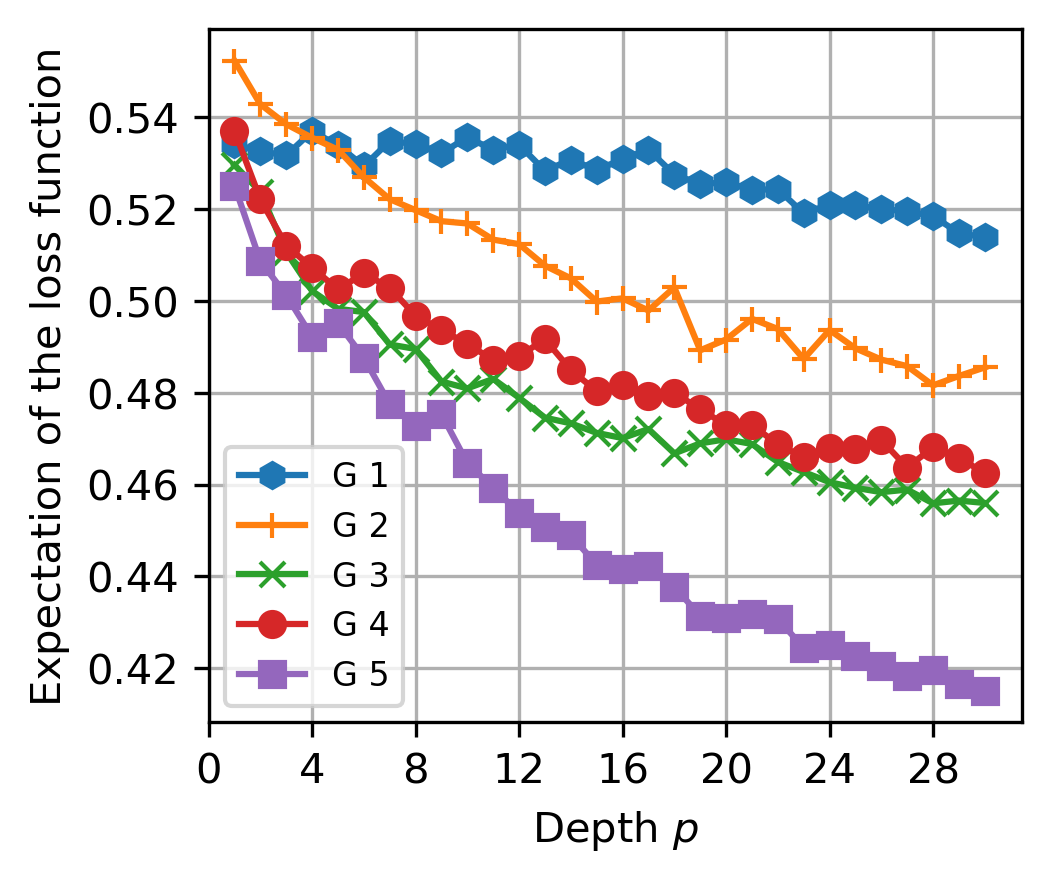}
\end{subfigure}
\hfill
\begin{subfigure}{}
%\caption{Variance vs. Depth.}\label{fig:var_depth} 
\includegraphics[width=0.49\textwidth]{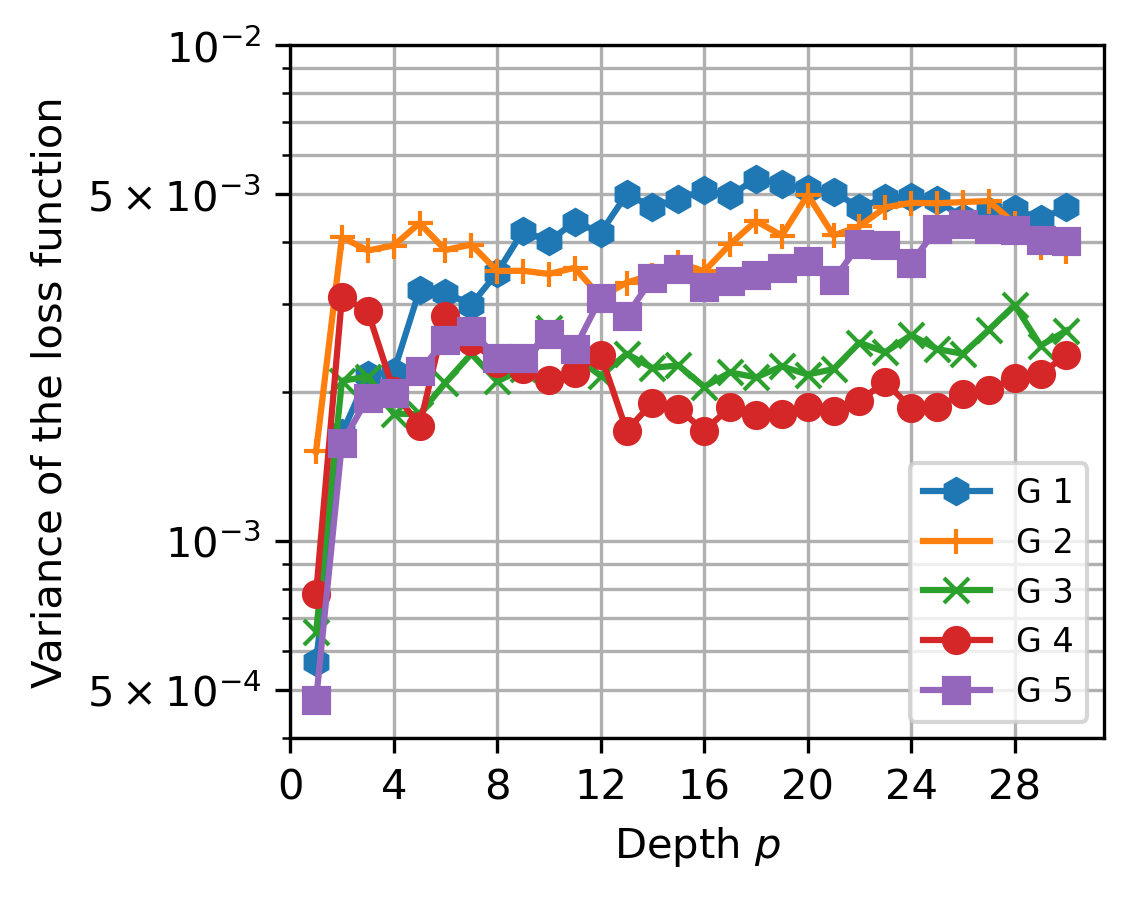}
\end{subfigure}
\caption{ 
    \textbf{Variance and expectation of the loss function across depths.} %For fixed $n=8$, we 
    We sampled five random graphs $G_i$ ($1\leq i \leq 5$) with $n=8$ vertices, and computed the expectation and variance of the loss function for $100$ random sampled parameters, with circuit depths $p$ from $1$ up to $30$. 
    \textbf{Left panel.} 
    %The expectation decays gradually across circuit depth, revealing that larger MaxCut values correspond to lower energy states.   
    As a function of $p$, the graph of the expectation seems convex up, which suggests it is a decreasing function with a slowing rate of decrease.
    \textbf{Right panel.} 
    %The variance grows initially for shallow circuits, after which steady convergence may indicate the circuit has reached a stable state. 
    The variance shows a sharp growth for small $p$, after which it flattens and stays elevated.
    The maximum and minimum values in this plot represent the bounds corresponding to $n=8$ in Fig.~\ref{fig:var_range}.
}
\label{fig:expvar_depth}
\end{figure*}

Our methodology centered on characterizing the statistical data (\ref{varianceFormula}) and (\ref{ExpectationOfLossFunction}), for GM-QAOA with MaxCut. 
Our findings for the expectation and variance of the loss function are presented in Figure \ref{fig:expvar_depth}.

Consistent with the notation in Theorem~\ref{thm:BPtheorem}, the normalized problem Hamiltonian we want to maximize is 
\[
\widehat{H}_P = \frac{1}{|E|}\sum_{(i,j)\in E} \frac{1}{2}(I - Z_iZ_j).
\] 
The exact Grover mixer $G_M$ is implemented as in ~\eqref{eq:GM_unitary}, where $U_\xi=H^{\otimes n}$. In PennyLane~\cite{pennylane}, this operation is realized by applying to the initial state $\ket{0\cdots 0}$ a sequence of gate operations consisting of: an initial layer of $n$ Pauli $X$ gates, followed by a multi-controlled phase gate which applies $e^{-i\beta}$ to an ancilla qubit, and a final layer of $n$ Pauli $X$ gates with the entire sequence conjugated by $H^{\otimes n}$. We further verified the correctness of this construction by confirming the consistency of the Grover mixing unitaries with the exact matrix representation in~\eqref{eq:GM} for small $n$.

\begin{figure}[h]
    \centering
    \includegraphics[width=0.45\textwidth]{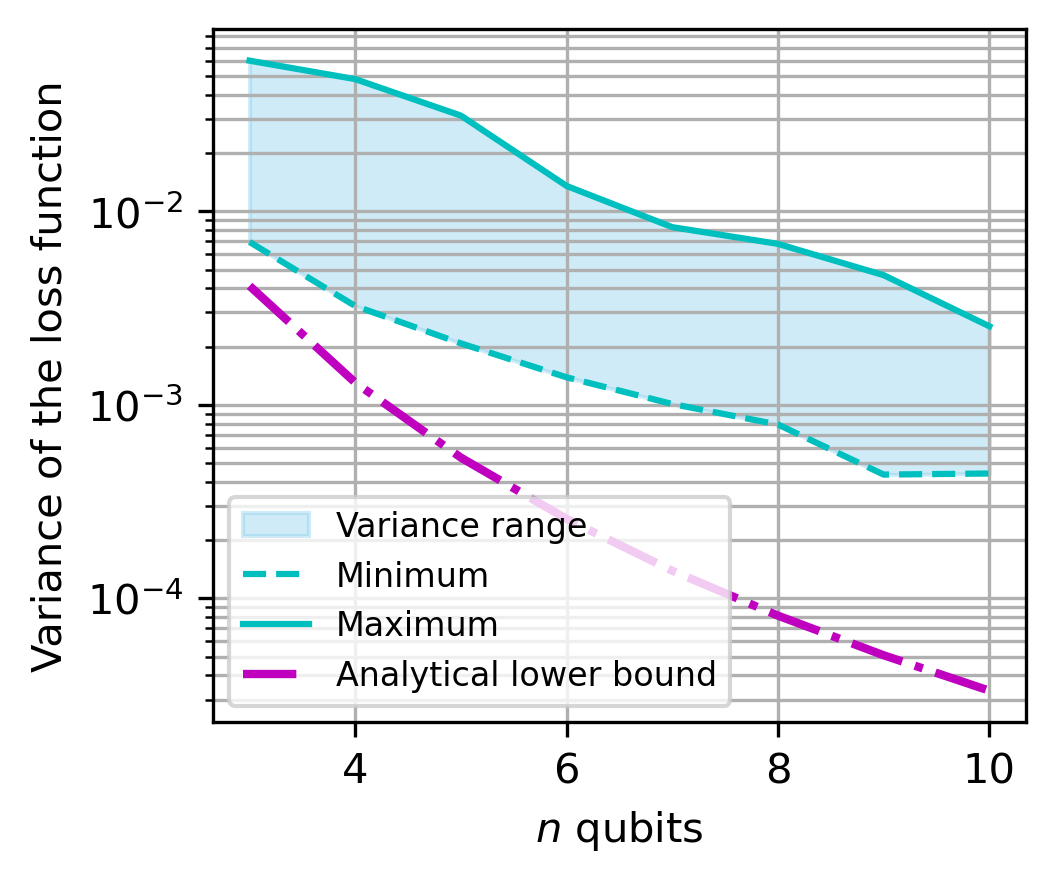} 
    \caption{
        \textbf{Upper and lower bounds for the variance of the loss function.} As a function of $n$, the variance is bounded from below by the analytical lower bound $\frac{1}{3n^4}$. For each $n$ up to $10$, we sampled $5$ random graphs, ran the QAOA circuit with randomly sampled $100$ parameters per graph, computed the variance of the loss function, and recorded the minimum and maximum variance for all depths up to $p=30$.
    }
    \label{fig:var_range}
\end{figure}

To verify the analytical lower bound $\frac{1}{3n^4}$ from Eq. \eqref{MaxCutVarBound}, 
we plotted in Figure \ref{fig:var_range} the empirical variance as a function of the number of qubits $n$. The variance values, obtained from the simulations, were compared directly against the theoretical curve to confirm that they remain above the predicted lower bound, consistent with Eq.\ \eqref{MaxCutVarBound}. Across all studies, we picked 5 randomly generated Erd\"os–R\'enyi graphs and sampled $100$ parameter pairs $(\boldsymbol{\beta}, \boldsymbol{\gamma})$, where $\beta_j\in [0,\pi)$ and $\gamma_j\in [0,2\pi)$.  

There were also important limitations. This numerical study did not address constrained optimization problems, a domain in which the Grover Mixer excels. Additionally, system sizes ($3\leq n \leq 10)$ and depths ($1\leq p \leq 30$) were limited by simulation cost, so extrapolation to larger-scale instances remains a challenge.

%Data generated and analyzed in this work can be reproduced from the code available at~\cite{maxcutqaoa}.

\subsection{Comparison of DLAs for GM-QAOA and Standard QAOA for MaxCut on Various Graphs}\label{sec:VI.B}

Unlike the dynamical Lie algebras $\mathfrak{g}_{\xi}$ arising from GM-QAOA, relatively little is known about their standard counterparts generated by the problem Hamiltonian together with the standard mixer 
\[
B = \sum_{j=1}^n X_j
\]
for the MaxCut problem. Below, we compare the known results for standard DLAs $\mathfrak{g}_{\mathrm{std}}$ 
with those for GM-QAOA DLAs  \( \mathfrak{g}_{+\cdots+} \).

\subsubsection{Path graphs} For MaxCut on the path graph with \( n \) vertices, the objective function \( F \) counts the number of edges crossing a cut. For any subset of the \( n-1 \) edges, there exists a vertex partition whose cut set contains exactly those edges.
 Therefore,
\(\operatorname{Range}(F) = \{0, 1, \dots, n-1\}\),
which implies

\[
\dim(\mathfrak{g}_{\xi}) = n^2 + 1.
\]
In comparison, it was shown that the standard DLA for the path graph \( \mathfrak{g}_{\mathrm{std}} \cong \mathfrak{u}(n) \) and has dimension 
\(n^2\), by Lemma~6 in \cite{KLFCCZ}.

\subsubsection{Cycle graphs} For MaxCut on the cycle graph with \( n \) vertices, label the vertices \( 1, 2, \dots, n \) consecutively around the cycle. Each cut corresponds to a partition of the vertex set \( V \) into two disjoint sets \( S \) and \( V \setminus S \). Without loss of generality, 
%assume \( |S| \le |V \setminus S| \), so 
it suffices to consider \( 0 \le |S| \le \lfloor \frac{n}{2} \rfloor+1 \).

The value of the objective function is given by the number of edges connecting vertices in \( S \) to those in \( V \setminus S \). Each vertex in \( S \) can contribute at most two edges to the cut, with the maximum value \( 2\lfloor \frac{n}{2} \rfloor \) achieved by an alternating configuration (for odd \( n \), a perfect alternating configuration is impossible, and two adjacent vertices reside within the same set). The minimum value is \( 0 \) when \( S = \varnothing \). 
%or \( S = V \). 
Only cuts with an even number of crossing edges can be realized, which implies that the objective function values are all even numbers between $0$ and $2\lfloor \frac{n}{2} \rfloor$.

 We conclude that the number of distinct values is \( r= \lfloor n/2 \rfloor + 1 \), leading to 
\[
\dim(\mathfrak{g}_{ \xi}) = \left(\lfloor n/2 \rfloor + 1\right)^2 + 1.
\]
In this case, the standard DLA \( \mathfrak{g}_{\mathrm{std}} \cong \mathfrak{su}(2)^{\oplus (n-1)} \oplus \mathfrak{u}(1)^{\oplus 2} \), by Theorem~3 in \cite{ASYZ1}. Hence, 
\[
\dim(\mathfrak{g}_{\mathrm{std}}) = 3n - 1.
\]
%which is smaller than \( \dim(\mathfrak{g}_{ \xi}) \).

\subsubsection{Complete graphs} In a complete graph, any subset of vertices of order \( k\) is connected to every vertex in its complement of size \( n-k\), so the number of edges crossing the cut is \( k(n-k) \). The number of distinct objective function values is \( r = \lfloor n/2 \rfloor + 1 \), giving 
\[
\dim(\mathfrak{g}_{ \xi}) %= r^2 + 1 
= \left(\lfloor n/2 \rfloor + 1\right)^2 + 1.
\]
%which grows quadratically in \( n \). 
The decomposition of the standard DLA $\mathfrak{g}_{\mathrm{std}}$ into simple components and its center is unknown in this case; however, it was shown that 
$\dim(\mathfrak{g}_{\mathrm{std}})$ grows cubically in \( n \) (see Section~4.3 in \cite{ASYZ1}).

\subsubsection{House graph} Finally, we provide an illustrative example that highlights the difference in dimensions between the two dynamical Lie algebras, based on our observations for the majority of graphs with up to six vertices.

Consider the house graph, a $5$-vertex graph obtained by joining a triangle and a square along a common edge (see Figure~\ref{HouseGraph}).  
A straightforward computation shows that the objective function takes exactly five distinct values.  
Hence $r=5$, and %we conclude that
\[
\dim(\mathfrak{g}_{ \xi}) %= 5^2+1
= 26.
\]
In contrast, for the standard DLA, % for the same graph has dimension
\[
\dim(\mathfrak{g}_{\mathrm{std}}) = 248,
\]
according to Figure~7 in \cite{KLFCCZ}.

\begin{figure}[h!]
\centering
\begin{tikzpicture}[scale=1]
  % Define vertices 1 and 2 (left side)
  \foreach \i/\ang in {1/135, 2/225} {
    \node[circle, draw, fill=black, inner sep=2pt, label={left:$\i$}] (\i) at (\ang:2) {};
  }
  % Define vertices 3 and 4 (right side)
  \foreach \i/\ang in {3/315, 4/45} {
    \node[circle, draw, fill=black, inner sep=2pt, label={right:$\i$}] (\i) at (\ang:2) {};
  }
  % Define vertex 5 (top)
  \node[circle, draw, fill=black, inner sep=2pt, label={above:$5$}] (5) at (0,3) {};
  
  % Draw edges
  \foreach \u/\v in {1/2, 2/3, 3/4, 4/1, 5/1, 5/4} {
    \draw (\u) -- (\v);
  }
\end{tikzpicture}
\caption{\textbf{House graph.}}
\label{HouseGraph}
\end{figure}
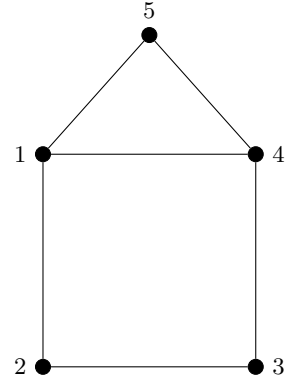

\subsection{Analyzing GM-QAOA for Optimization Problems beyond MaxCut}\label{sec:VI.C}

% We now demonstrate how Theorems~\ref{DLAMainThm} and~\ref{LossFuncVariance} apply to several important classes of combinatorial optimization problems, highlighting %both 
% the broad relevance %and the practical implications 
% of our results.

Beyond the MaxCut problem analyzed in the preceding subsection, we now present several additional examples of combinatorial optimization problems with \( s \)-local objective functions. These illustrate the wide applicability of GM-QAOA and offer further insight into the structure of the associated dynamical Lie algebras.

%\begin{ex}
%\begin{enumerate}

%        \item For vertex or edge coloring problems in \( c \) colors (encoded using the $c$-dit representation), the domain has cardinality \( c^n \), where \( n \) is the number of vertices. The objective function counts the number of violations (i.e., adjacent vertices having the same color), so its range is bounded below by \( 0 \) and above by \( |E| \), the number of edges. Therefore, there are at most \( |E| + 1 \) distinct objective values, and the variance satisfies \(\operatorname{Var}_{\boldsymbol{\beta}, \boldsymbol{\gamma}}\left[\ell_{\boldsymbol{\beta}, \boldsymbol{\gamma}} (\rho, H_P) \right]=\Omega\left(\frac{1}{n^{4}}\right) \).
        
        \subsubsection{The $m$-SAT problem} 
        The $m$-SAT problem asks whether there exists a bit string  \( x = (x_1, \dots, x_n) \in \mathbb{B}^n \) that satisfies a given collection of clauses, each of which is a disjunction of \( m \) literals (variables or their negations). We can encode this as an $m$-local objective function of the form \eqref{sLocalFunc},
        % \[
        %     F(x_1, \dots, x_n) = \sum_{j=1}^{T} f_j(x_{S_j}),
        % \]
        where each clause function \( f_j \colon \mathbb{B}^m \to \{0,1\} \) evaluates to 1 if the clause is satisfied and 0 otherwise, and \( S_j \subseteq \{1, \dots, n\} \) is the subset of variables appearing in the $j$-th clause. 

        Applying Theorem~\ref{thm:BPtheorem} with $s=m$ and $M=1$, we obtain a lower bound on the variance of the loss function:
        \[
            \operatorname{Var}_{\boldsymbol{\beta}, \boldsymbol{\gamma}}\bigl[\ell_{\boldsymbol{\beta}, \boldsymbol{\gamma}} (\rho, \widehat{H}_P)\bigr] 
            \ge \frac{(m!)^2}{12 n^{2m}} \,,
        \]
        for sufficiently large circuit depth $p$.

        \subsubsection{Max-$k$-VertexCover} 
        Max-$k$-VertexCover is a constrained optimization problem on an undirected graph $G=(V,E)$ with $n$ vertices. The goal is to select a subset of $k$ vertices such that the number of edges incident to the chosen vertices is maximized. The objective function can be written as
        \[
            F(x) = \sum_{(i,j)\in E} \bigl(1 - (1-x_i)(1-x_j)\bigr),
        \]
        where $x_i = 1$ if vertex $i$ is included in the vertex cover and $x_i = 0$ otherwise. This function counts the number of edges covered by the selected vertices.

The constraint that we have selected exactly $k$ vertices restricts $x$ to bit strings of Hamming weight $k$. We take the initial state $\ket\xi$ to be their uniform superposition, i.e., the Dicke state $\ket{D^n_k}$.
       Note that $F(x)$ is $2$-local, and each constituent term $1 - (1-x_i)(1-x_j)$ attains a maximal value $M=2$. Therefore, by Theorem~\ref{thm:BPtheorem}, the variance of the corresponding loss function for sufficiently large circuit depth $p$ is bounded below by
        \(
             \frac{1}{12 n^{4}}.
        \)
        
        %We refer the reader to \cite{XXCLSA} and \cite{BE} for numerical simulating the performance of GM-QAOA for Max-$k$-VertexCover.
        
         \subsubsection{Traveling Salesperson Problem} 
         The Traveling Salesperson Problem (TSP) is a permutation based optimization problem where the goal is to find an ordering of $k$ cities that minimizes the distance traveled if we visit each city once and return to the origin.
         We label the cities $1,\dots,k$, and
         represent a feasible tour as a permutation $\tau =(v_1,\dots,v_k)$,
         %on the cities labeled by $v_i$ such that $(v_i,v_{i+1})\in E$, 
         indicating traveling from city $v_1$ to city $v_2$, then $v_2$ to $v_3$, and so on until $v_n$, and returning to $v_1$. 
         
         Let the positive integer $w_{ij} = w_{ji}$ be the distance from city $i$ to city $j\ne i$.
         %, i.e., the minimum number of edges that create a path from $i$ to $j$. 
         Then our goal is to minimize the objective function value of a tour $\tau=(v_1,\dots, v_k)$, 
         \[
         F(\tau)=\sum_{j=1}^k w_{v_jv_{j+1}}, \quad\text{where}\;\; v_{k+1}:=v_1.
         \]
         Setting $w_{\mathrm{min}}= \min\limits_{i\neq j}w_{ij}$ and $w_{\mathrm{max}}=\max\limits_{i\neq j}w_{ij}$, gives us possible objective function values between $kw_{\mathrm{min}}$ and $kw_{\mathrm{max}}$.
         
         Similar to the constructions in~\cite{Smelyanskiy2012NearTerm} and~\cite{HWORVB}, we can identify TSP with a $2$-local objective function,
        \[
            F(x)= \frac{1}{2}\sum_{1\le q,i,j\le k}w_{ij}x_{i,q}(x_{j,q+1} + x_{j,q-1}),
        \]
        where $x_{i,q}=1$ if city $i$ is in the $q$-th position in the tour and 0 otherwise. The additional constraints $\sum\limits_{q=1}^k x_{i,q}=1$ and $\sum\limits_{i=1}^k x_{i,q}=1$ ensure that each city is visited exactly once. Note this formalization of the problem requires $n=k^2$ qubits to encode  vectors $x=(x_{i,q})_{1\leq i,q\leq n}$.
        
        Applying Theorem~\ref{thm:BPtheorem}  with $n=k^2$, $s=2$ and $M= w_{\mathrm{max}}$, we obtain  $$\operatorname{Var}_{\boldsymbol{\beta}, \boldsymbol{\gamma}}\bigl[\ell_{\boldsymbol{\beta}, \boldsymbol{\gamma}} (\rho, \widehat{H}_P) \bigr] \geq \frac{1}{3w_{\mathrm{max}}^2k^{8}}.$$
         %The preparation of $U_\xi$ requires approximately $n^2$ qubits for encoding, and its evaluation complexity is $\mathcal{O}(n^3)$ \cite{KHR}.
%\end{enumerate}
%\end{ex}

\section{Discussion}\label{sec:VII}

Dynamical Lie algebras (DLAs) have emerged as a promising tool for analyzing the performance of variational quantum algorithms. 
Their structural properties---in particular, the dimensions of the simple subalgebras arising in the decomposition of the semisimple part into simple constituents---directly influence the variance of the loss function and thereby play a central role in determining the presence or absence of barren plateaus \cite{RBSKMLC,FHCKYHSP}. Despite this conceptual significance, the practical applicability of DLAs has been limited so far, as these algebras are in general difficult to compute explicitly.

Sharp contrasts arise when one compares different ansätze. 
For example, in the case of the multi-angle QAOA \cite{herrman2022multi} or also called free ansatz---where the generating set consists of all summands in both mixer and problem Hamiltonians---the associated Lie algebras are completely classified \cite{WKKB2} (see also \cite{WKKB1, KLFCCZ}). 
By contrast, the precise structure of the DLA associated to the standard QAOA ansatz, which uses the standard $X$-mixer \eqref{Xmixer}, is generally unknown (see \cite{ASYZ1, KLFCCZ} for partial results).

In the present manuscript, we provide a complete classification of the dynamical Lie algebras arising in QAOA with Grover-type mixers (defined by projections onto the initial state vector). 
The GM-QAOA DLAs admit a transparent decomposition of the Hilbert space into an irreducible subspace $W_0$ (spanned by the nonzero projections of the initial state onto distinct objective-value eigenspaces) plus a direct sum of one-dimensional invariant subspaces. %where the subspace containing the superposition of computational basis states share coinciding eigenvalues with respect to the problem Hamiltonian.
Therefore, the dynamics of the GM-QAOA circuit are confined to $W_0$ and can approximate the superposition over optimal solutions for large enough depth.

If we consider the GM-QAOA DLA $\mathfrak{g}_{\xi}$ and any other QAOA DLA $\mathfrak{g}$ with a mixer Hamiltonian initialized in the same way, then we proved that the commutant $\mathcal{C}(\mathfrak{g}_{\xi})$ contains $\mathcal{C}(\mathfrak{g})$, while the associative algebra $\mathcal{A}(\mathfrak{g}_{\xi})$, generated by $\mathfrak{g}_{\xi}$, is contained in $\mathcal{A}(\mathfrak{g})$. 
This means that GM-QAOA attains the largest possible commutant among mixers with the same initial state, so it maximally preserves symmetries of the system. The inclusion of associative algebras tells us that the irreducible representation $W_0$ of the GM-QAOA DLA $\mathfrak{g}_{\xi}$ sits inside the corresponding irreducible representation of $\mathfrak{g}$.
Therefore, any other QAOA also reaches an optimal solution for large enough depth.
%active subspace of 
%any other QAOA DLA $\mathfrak{g}$.
%whose mixer with the same initial state.

Furthermore, we demonstrate that, for a large class of optimization problems, the GM-QAOA DLA has dimension bounded by a polynomial in the problem size. 
This observation leads to an important algorithmic consequence: it rules out the emergence of barren plateaus, provided the depth $p$ is sufficiently large \cite{RBSKMLC}.
Determining explicit lower bounds on the required $p$ remains an open problem, which we leave for future work.

Since Grover-type mixers exhibit favorable behavior with respect to barren plateaus, it may be of interest to incorporate them into the pool of mixer candidates considered within Adapt-QAOA (introduced in ~\cite{ZTBCEBM}). 
In this setting, the adaptive strategy could select Grover mixers whenever their algebraic structure leads to enhanced trainability, thus potentially improving the efficiency and robustness of the algorithm.

\section*{Code Availability}
Data generated and analyzed in this work can be reproduced from the code available at~\cite{maxcutqaoa}.

\section*{Acknowledgments}
BNB acknowledges valuable discussions with Eleanor Rieffel, Michael Ragone, and Manas Sajjan.
We thank the authors of~\cite{BLSAMS} for bringing their results to our attention and for pointing out typos in the first version of our manuscript.
BNB and MN were supported by the U.S. Department of Energy, Advanced Scientific Computing Research, under contract number DE-SC0025384. 

\bibliographystyle{plain}
\bibliography{referencesMain}

\appendix
\section{Proofs and Technical Details}\label{sec:VIII}

In this appendix, we provide detailed and rigorous proofs of the results stated in the main text. Our goal is to make the exposition self-contained and to supply all intermediate steps needed for verification.

\subsection{Proofs of Theorems \ref{DLAMainThm} and \ref{IsotypicDecomp}}\label{sec:proofIII.1}
Recall the decomposition $W=\widetilde{W} \oplus W_0$
% \[
%     W \;=\; \bigoplus_{j=1}^{d} \Bigl( \widetilde{W}_{\lambda_j} \oplus \mathbb{C}\ket{\xi_j} \Bigr) 
%     \;=\; \widetilde{W} \oplus W_0,
% \]
from Eq.~\eqref{simplifiedDecompForGmax}, where $\widetilde{W}$ and $W_0$ are defined in \eqref{eq:Wtilde}, \eqref{eq:W0}.
The actions of $G_M$ and $H_P$ on $W_0$ are given by Eqs.\ \eqref{RestrictionsToTriv2} and \eqref{RestrictionsToTriv1}.
% Here, for each $j$, 
% \[
%     \widetilde{W}_{\lambda_j} := W_{\lambda_j} \cap \bigl(\mathbb{C}\ket{\xi}\bigr)^{\perp}
% \]
% is the subspace of $W_{\lambda_j}$ orthogonal to the initial vector $\ket\xi$.  The dimension of $\widetilde{W}_{\lambda_j}$ is equal to $n_j - 1$,
% where $n_j$ denotes the number of elements in the preimage of $\lambda_j$ under the objective function $F$, or equivalently, the multiplicity of the eigenvalue $\lambda_j$ of the problem Hamiltonian $H_P$.
%
% By construction, the generators of the Lie algebra \( \mathfrak{g}_{ \xi} \) act on \( \widetilde{W}\) in the following way.
% \begin{itemize}
%     \item The operator \( G_{M}  \), which stabilizes \( \ket{\xi}\), acts trivially on \( \widetilde{W}\).
%     \item The operator \( H_P \), being diagonal in the computational basis, leaves each subspace 
% \( \widetilde{W}_{\lambda_j} \) and \( W_{\lambda_j} \) invariant, and acts on it as scalar multiplication by \( \lambda_j \).
% \end{itemize}

Note that $G_M =  \ket{\xi}\bra{\xi}$ acts trivially on $\widetilde{W}$ because $\widetilde{W} \subset \bigl(\mathbb{C}\ket{\xi}\bigr)^{\perp}$.
On the other hand, $H_P$ acts as multiplication by \( \lambda_j \) on its eigenspaces \( W_{\lambda_j} \) and on their subspaces \( \widetilde{W}_{\lambda_j} \subset W_{\lambda_j} \).
We obtain the following:

\begin{lem}\label{lem:Wtilde}
The subspace $\widetilde{W}$ is a direct sum of $1$-dimensional representations of the DLA \( \mathfrak{g}_{ \xi} \).
\end{lem}

%the space \( \bigoplus\limits_{j=1}^k \widetilde{W}_{\lambda_j} \) decomposes into 1-dimensional invariant subspaces with respect to the action of \( H_P \).

% The actions of \( H_P \) and \( G_{M}  \) on $W_0$ can be described as follows.
% In the orthonormal basis \( \{ \ket{\xi_1}, \dots, \ket{\xi_d} \} \), they are represented by the matrices:
% %the restrictions of the cost Hamiltonian \( H_P \) and the projector \( G_{M}  \) to the  subspace \( W_{0} \) take the form:
% \begin{align}\label{RestrictionsToTriv1}
% H_{P, 0} &:=
% \begin{pmatrix}
% \lambda_1 & 0 & \cdots & 0 \\
% 0 & \lambda_2 & \cdots & 0 \\
% \vdots & \vdots & \ddots & \vdots \\
% 0 & 0 & \cdots & \lambda_{d}
% \end{pmatrix},\\
% G_{M, 0} &:=
% -\begin{pmatrix}
% |c_1|^2 & c_1 \bar{c}_2 & \cdots & c_1 \bar{c}_d \\
% c_2 \bar{c}_1 & |c_2|^2 & \cdots & c_2 \bar{c}_{d} \\
% \vdots & \vdots & \ddots & \vdots \\
% c_{d} \bar{c}_1 & c_{d} \bar{c}_2 & \cdots & |c_d|^2
% \end{pmatrix},
% \label{RestrictionsToTriv2}
% \end{align}
% due to Eq.~\eqref{InitStateDecomp} and $\braket{\xi_j}{\xi} =c_j$.

% These expressions follow from the decomposition of the initial state \( \ket{\xi}\) in terms of \( \ket{\xi_1},\dots,\ket{\xi_d}\), as described in Eq.~\eqref{InitStateDecomp}, with inner products given by
% \[
% \braket{\xi_j}{\xi} =c_j.
% \]
%\begin{rmk}
% Note that the restriction of the mixer Hamiltonian can be written as
% \[
%     G_{M, 0} \;=\; - \ket{\xi}\bra{\xi}.
% \]
%\end{rmk}
Next, we present a technical result that serves as the main ingredient for describing dynamical Lie algebras arising from GM-QAOAs. 
Although our focus is on Grover-mixer QAOAs, the result may also be applicable to a broader class of QAOAs. %originating from variational quantum algorithms.

\begin{prop}
Let \( \mathcal{D} = \mathrm{diag}(\lambda_1, \dots, \lambda_{d}) \) be a diagonal real matrix with %pairwise distinct entries satisfying 
\( \lambda_1 > \cdots > \lambda_{d} \), and \( A \in \mathbb{C}^{d \times d} \) be any complex matrix such that its entries satisfy
%for all indices \( j \in \{2, \dots, d-1\} \), we have
\begin{align*}
&a_{i1} a_{1j} + a_{id} a_{dj} \ne 0, \quad 2\le i,j \le d-1, \\
&a_{1k} \ne 0, \, a_{k1} \ne 0, \, a_{kd} \ne 0, \, a_{dk} \ne 0, \; 1\le k \le d.
\end{align*}
Consider the complex Lie algebra \( \mathfrak{p} := \langle \mathcal{D}, A \rangle_{\mathrm{Lie},\mathbb{C}} \) generated by \( \mathcal{D} \) and \( A \).
If both \( \mathcal{D} \) and \( A \) are traceless, then \( \mathfrak{p} = \mathfrak{sl}(d,\mathbb{C}) \); otherwise, \( \mathfrak{p} = \mathfrak{gl}(d,\mathbb{C}) \).
\label{KeyProp}
\end{prop}

\begin{proof}
Recall the commutation relations
\begin{equation}\label{EijEkl}
[E_{ij}, E_{kl}] = \delta_{jk} E_{il} - \delta_{li} E_{kj},
\end{equation}
where $E_{ij}$ is the matrix with $(i,j)$-entry $1$ and all other entries $0$.
Note also that
\begin{equation}\label{aijEij}
A = \sum_{i,j=1}^d a_{ij} E_{ij}.
\end{equation}
From \eqref{EijEkl}, we derive
$[\mathcal{D}, E_{ij}] = (\lambda_i - \lambda_j) E_{ij}$.
Hence, for any polynomial $g(t)$, we have 
\[
g(\operatorname{ad}_{\mathcal{D}}) E_{ij} = g(\lambda_i - \lambda_j) E_{ij}.
\]

Denote by $\Lambda := \{ \lambda_i - \lambda_j \mid 1\le i,j \le d\}$ the set of all pairwise differences of the eigenvalues of $\mathcal{D}$. For each pair of indices $1 \leq i \ne j \leq d$, let \( g_{ij}(t) \in \mathbb{C}[t] \) be the Lagrange interpolation polynomial such that
\[
g_{ij}(t) =
\begin{cases}
1 & \text{if }\, t = \lambda_i - \lambda_j, \\
0 & \text{if }\, t \in \Lambda \setminus \{\lambda_i - \lambda_j\}.
\end{cases}
\]

%\textbf{Step 1.} 
The largest difference \( \lambda_1 - \lambda_{d} \) occurs uniquely in $\Lambda$, due to the strict ordering of the entries in \( \mathcal{D} \). 
This implies that
\begin{align*}
g_{1d}(\operatorname{ad}_{\mathcal{D}})(A) &= a_{1d} E_{1d}, \\
g_{d1}(\operatorname{ad}_{\mathcal{D}})(A) &= a_{d1} E_{d1}.
\end{align*}
Because, by assumption, $a_{1d} \ne 0$ and $a_{d1} \ne 0$, we obtain that \(E_{1d},E_{d1}\in\mathfrak{p}\).
Then also $E_{11}-E_{dd} = [E_{1d},E_{d1}] \in\mathfrak{p}$.

Now consider the matrix $[E_{1d},A] \in\mathfrak{p}$, which by Eqs. \eqref{EijEkl}, \eqref{aijEij} is
\[
[E_{1d},A] = \sum_{j=1}^d a_{dj} E_{1j} - \sum_{i=1}^d a_{i1} E_{id}.
\]
If we apply $g_{1j}(\operatorname{ad}_{\mathcal{D}})$ to it, for a fixed $2 \le j \le d-1$, only the term with $E_{1j}$ contributes from the first sum, because $\lambda_1 - \lambda_j$ are pairwise distinct. The second sum may contribute at most one term corresponding to $E_{id}$ with $\lambda_1 - \lambda_j = \lambda_i - \lambda_d$.
We obtain that:
\[
g_{1j}(\operatorname{ad}_{\mathcal{D}})([E_{1d},A]) = a_{dj} E_{1j},
\]
if $\lambda_1 - \lambda_j \ne \lambda_i - \lambda_d$ for all $i$, and
\[
g_{1j}(\operatorname{ad}_{\mathcal{D}})([E_{1d},A]) = a_{dj} E_{1j} - a_{i1} E_{id},
\]
if $\lambda_1 - \lambda_j = \lambda_i - \lambda_d$.

In the former case, we conclude that $E_{1j}\in\mathfrak{p}$. In the latter case, we further consider
\begin{align*}
[E_{11}-E_{dd},A] &= \sum_{j=1}^d a_{1j} E_{1j} + \sum_{i=1}^d a_{id} E_{id} \\
&- \sum_{i=1}^d a_{i1} E_{i1} - \sum_{j=1}^d a_{dj} E_{dj}.
\end{align*}
Then, as above, we get
\[
g_{1j}(\operatorname{ad}_{\mathcal{D}})([E_{11}-E_{dd},A]) = a_{1j} E_{1j} + a_{id} E_{id}.
\]
Using the assumption $a_{i1} a_{1j} + a_{id} a_{dj} \ne 0$, we conclude again that $E_{1j}\in\mathfrak{p}$.

A similar argument, using
\[
[E_{d1},A] = \sum_{j=1}^d a_{1j} E_{dj} - \sum_{i=1}^d a_{id} E_{i1},
\]
proves that $E_{i1}\in\mathfrak{p}$.

Now the identity
\[
E_{ij} = [E_{i1}, E_{1j}], \qquad 2 \leq i \ne j \leq d,
\]
shows that all off-diagonal matrices lie in \( \mathfrak{p} \). 
%Moreover, the subspace of diagonal matrices is spanned by $\mathcal{D}$ and the traceless diagonal matrices
Then from
\[
E_{ii} - E_{jj} = [E_{ij}, E_{ji}],
\]
we see that
\(
\mathfrak{p} \supseteq \mathfrak{sl}(d,\mathbb{C}). 
\)
Finally, if \( \mathcal{D} \) or \( A - \sum\limits_{i \ne j} a_{ij} E_{ij} \) is a diagonal matrix with nonzero trace, then \( \mathfrak{p} = \mathfrak{gl}(d,\mathbb{C}) \).
\end{proof}

\begin{cor}
Let 
\(
\mathcal{D} = \mathrm{diag}(\lambda_1, \dots, \lambda_{d})
\) 
be a real diagonal matrix with %pairwise distinct entries satisfying 
\(\lambda_1 > \cdots > \lambda_{d}\), and let 
\(A \in i\mathfrak{u}(d)\) be a Hermitian matrix with nonzero trace, 
satisfying the assumptions of Proposition~\ref{KeyProp}.
% such that for all indices \( j \in \{2, \dots, d-1\} \),
% \[
% a_{1j} \ne 0, \quad a_{j1} \ne 0, \quad a_{jd} \ne 0, \quad a_{dj} \ne 0.
% \]
Then the Lie algebra
\(
\langle i\mathcal{D}, iA \rangle_{\mathrm{Lie}} = \mathfrak{u}(d)
%\text{ or } \mathfrak{su}(d)
\).
%Its complexification is 
% \(
%     \langle \mathcal{D},\, A \rangle_{\mathrm{Lie}, \mathbb{C}},
% \) 
% which, by Proposition~\ref{KeyProp}, is either $\mathfrak{sl}(d,\mathbb{C})$ or $\mathfrak{gl}(d,\mathbb{C})$.  
% Multiplying the real symmetric generators $\mathcal{D}$ and $A$ by $i$ produces skew-Hermitian matrices, so $\mathfrak{p}'$ is the compact real form of $\langle \mathcal{D}, A \rangle_{\mathrm{Lie}, \mathbb{C}}$. Hence, following Proposition~\ref{KeyProp} and using the fact that $\mathfrak{p}'$ is reductive (see Proposition A.1 in \cite{WKKB2}), we obtain:
% \begin{enumerate}
%     \item if both $\mathcal{D}$ and $A$ are traceless, $\mathfrak{p}'$ is isomorphic to $\mathfrak{su}(d)$, the compact real form of $\mathfrak{sl}(d,\mathbb{C})$;
%     \item otherwise, $\mathfrak{p}' \cong \mathfrak{u}(d)\cong \mathfrak{su}(d)\oplus \mathfrak{u}(1)$, the compact real form of $\mathfrak{gl}(d,\mathbb{C})$.
% \end{enumerate}

\label{KeyCor}
\end{cor}

%\begin{rmk}
 %  It is natural to refer to the condition 
  %  \[ a_{1j} \ne 0, \quad a_{j1} \ne 0, \quad a_{jd} \ne 0, \quad a_{dj} \ne 0, \]
   % for all \(1 \le j\le d\), as the \textit{frame condition}, since these entries lie along the top row, bottom row, leftmost column, and rightmost column of the matrix, but exclude the corners. That is, they form a quadratic ``frame" around the matrix with the corner entries removed. This structure is illustrated schematically below for a generic \( d \times d \) matrix:
 %     \[ A = \begin{pmatrix}
  %  \star & \star & \cdots & \star & \star\\
 %   \star & \bullet & \cdots & \bullet & \star \\
%    \vdots & \vdots & \ddots & \vdots & \vdots \\
 %   \star & \bullet & \cdots & \bullet & \star \\
%    \star & \star & \cdots & \star & \star
 %   \end{pmatrix}. \]
 %   The \( \star \) entries indicate where nonzero values are required under the frame condition. 
%\end{rmk}

We are now ready to present the proof of Theorem~\ref{DLAMainThm}, which characterizes the dynamical Lie algebras for QAOAs employing \( G_{M} \) as the mixer Hamiltonian.

\begin{proof}[Proof of Theorem~\ref{DLAMainThm}]
The Lie algebra \(\mathfrak{g}_{ \xi}\) preserves the decomposition
\[
    W = \widetilde{W} \oplus W_{0} 
\]
(see~\eqref{simplifiedDecompForGmax}), because both generators $G_M$ and $H_P$ do.
Consequently, \(\mathfrak{g}_{ \xi}\) embeds into the direct sum of the corresponding Lie algebras:
\[
\mathfrak{g}_{ \xi} \subseteq \mathfrak{u}(\widetilde{W}) \oplus \mathfrak{u}(W_{0}).
\]

The restrictions $G_{M,0}$ and \( H_{P,0} \) of $G_M$ and $H_P$ to \( W_{0} \), expressed in the orthonormal basis 
\(
    \{ \ket{\xi_1}, \dots, \ket{\xi_d} \},
\)
are given explicitly by Eqs.\ \eqref{RestrictionsToTriv2}, \eqref{RestrictionsToTriv1}. As
\[
c_{i}\overline{c}_1 c_{1}\overline{c}_j + c_{i}\overline{c}_d c_{d}\overline{c}_j 
    = (|c_1|^2 + |c_d|^2)c_i \overline{c}_j \ne 0,
\]
the matrices \( H_{P,0} \) and \( G_{M,0} \) satisfy the hypotheses of Proposition~\ref{KeyProp}. Moreover,
\[
\mathrm{Tr}[G_{M,0}] = \sum_{j=1}^{d} |c_j|^2 = \braket{\xi}{\xi} = 1 \ne 0.
\]
By Corollary \ref{KeyCor}, the subalgebra of $\mathfrak{g}_{ \xi}$ generated by these restrictions is equal to 
\(
    \mathfrak{u}(W_0) \cong \mathfrak{u}(d).
\)

If $d=2^n$, then $\widetilde{W} = \{0\}$ and we are left with $\mathfrak{g}_{\xi} = \mathfrak{u}(2^n)$. 
Suppose now that \(d<2^n\). By definition, all elements of $\mathfrak{g}_{\xi}$ have the form
$\alpha iH_P + \beta i G_M +$ terms involving commutators of $H_P$ and $G_M$ (for some $\alpha,\beta\in\mathbb{R}$).
Since \( G_{M} \) vanishes on \( \widetilde{W} \),
we see that the projection of $\mathfrak{g}_{\xi}$ to the summand $\mathfrak{u}(\widetilde{W})$ is equal to
$\mathbb{R} iH_P \pi_{\widetilde{W}}$.
\end{proof}

%
%(since \(\braket{\xi}{w} = 0\) for all \( w \in \widetilde{W} \)). Therefore, both the derived subalgebra 
% \(
%     [\mathfrak{g}_{ \xi}, \mathfrak{g}_{ \xi}]
% \)
% and the operator \( G_{M} \) are contained entirely within \( \mathfrak{u}(W_{0}) \).

% We first assume that \(d<2^n\). 

% It remains to observe that any element of \( \mathfrak{g}_{ \xi}\) with nontrivial action on \( \widetilde{W} \) can be written in the form \( \alpha iH_P + a \), where \( \alpha \neq 0 \) and \( a \in \mathfrak{g}_{ \xi} \) acts trivially on \( \widetilde{W} \). In particular, the one-dimensional central abelian subalgebra appearing in the statement of the theorem is spanned by \( H_P - a \), with 
% \(a = iH_{P,0} \in \mathfrak{u}(d)\).
%
% We conclude that
% %the former two cases stated in Theorem \ref{DLAMainThm}
% \[
%     \mathfrak{g}_{ \xi} \cong \mathfrak{su}(d) \oplus \mathfrak{u}(1)^{\oplus 2}. 
% \]
%
% Finally, in the case where all eigenspaces of $H_P$ are one-dimensional and $d = 2^n$, 
% we have $\widetilde{W} = 0$ and $W = W_0$. 
% It then follows from Corollary~\ref{KeyCor} that 
% \[
%     \mathfrak{g}_{ \xi} \cong \mathfrak{su}(2^n) \oplus \mathfrak{u}(1),
% \] 
%thus confirming the last statement of Theorem~\ref{DLAMainThm}.
%\end{proof}

From the above proof, we also obtain the description of the center of $\mathfrak{g}_{\xi}$ given in Corollary \ref{cor:center}. Another consequence is that the representation of $\mathfrak{g}_{\xi}$ on $W_0$ is irreducible,
because $W_0$ is irreducible under the action of $\mathfrak{u}(W_0)$. Combined with Lemma \ref{lem:Wtilde}, 
this proves Theorem~\ref{IsotypicDecomp}.

Finally, again from the above proof of Theorem~\ref{DLAMainThm}, we have
\begin{equation}\label{eq:restate}
\mathfrak{g}_{\xi} = \mathfrak{u}(W_0) \pi_{W_0} \oplus \mathbb{R} iH_P \pi_{\widetilde{W}}.
\end{equation}
The associative closure of this gives Corollary \ref{cor:Agxi}.

% The \emph{center} of the dynamical Lie algebra, 
% \(
% \mathfrak{z} \subset \mathfrak{g}_{ \xi},
% \) 
% is defined as the set of all elements that commute with every element of $\mathfrak{g}_{ \xi}$. Equivalently, it consists of all elements that commute with both $G_M$ and $H_P$. 

% \begin{cor}\label{Center}
% Under the assumptions of Theorem~\ref{DLAMainThm}, if $d < 2^n$, the center $\mathfrak{z}$ of the Lie algebra $\mathfrak{g}_{\xi}$ is $2$-dimensional. It is spanned by the restriction of $H_P$ to $\widetilde{W}$ and by the identity operator on $W_0$.
% \end{cor}

\subsection{Proof of Theorem \ref{ContainmentThm}}

By assumption $H_M\ket\xi=\mu\ket\xi$, for some $\mu\in\mathbb R$, 
and all other eigenvectors of $H_M$ have eigenvalues $\ne\mu$.
Let $g(t)$ be the Lagrange interpolation polynomial such that
\[
g(t) =
\begin{cases}
1 & \text{if }\, t = \mu, \\
0 & \text{if }\, t \ne\mu \;\text{is an eigenvalue of $H_M$}.
\end{cases}
\]
Then $g(H_M)\ket\xi = \ket\xi$ and $g(H_M)=0$ on all other eigenvectors of $H_M$. Because eigenvectors with different eigenvalues are orthogonal, we see that $g(H_M)=0$ on the orthogonal complement of $\ket\xi$. Therefore, 
$g(H_M)=\ket\xi\bra\xi = G_M$,
completing the proof of Theorem \ref{ContainmentThm}.

% To compare the expressive power of GM-QAOA with that of the standard mixer QAOA, it is convenient to work with the associative algebras generated by the corresponding DLAs.

% \begin{defn}
% For a Lie algebra \(\mathfrak{g}\) acting on the \(n\)-qubit Hilbert space 
% \(W = \mathbb{C}^{2^{n}}\), we denote by 
% \(\mathcal{A}_{\mathfrak{g}} \subseteq \operatorname{End}(W)\) 
% the unital associative \(\mathbb{C}\)-algebra generated by the image of \(\mathfrak{g}\) 
% in \(\operatorname{End}(W)\).
% \end{defn}

%Theorem \ref{ContainmentThm} is an immediate consequence of the following result.
In the case of the standard mixer \(B=\sum\limits_{j=1}^n X_j\), we also have the following result,
which provides an alternative proof that the Grover mixer $G_M$ is a polynomial of $B$ for $\ket{\xi}=\ket{+\cdots+}$
(see Eq.~\eqref{eq:grover}).

\begin{prop}\label{ContainmentAlgThm}
% Let $\ket{\xi}=\ket{+\cdots+}$ and \(\mathfrak{g}_{\mathrm{std}}:=\langle iB,\; iH_P\rangle_{\mathrm{Lie}}\) be the DLA for QAOA with the standard mixer \(B=\sum\limits_{j=1}^n X_j\).  Then there is an inclusion of associative algebras
% \(
%     \mathcal{A}_{\mathfrak{g}_\xi} \;\subseteq\; \mathcal{A}_{\mathfrak{g}_{\mathrm{std}}}.
% \)
%The Grover mixer $G_M$ is a polynomial of $B$.
The elementary symmetric polynomials \eqref{eq:el} can be expressed as polynomials of $B$.
\end{prop}

\begin{proof}
% By definition, $\mathcal{A}_{\mathfrak{g}_{\mathrm{std}}}$ contains $H_P$, which is one of the two generators of $\mathcal{A}_{\mathfrak{g}_\xi}$.  
% Thus, it remains to show that the second generator, $G_M$, also lies in $\mathcal{A}_{\mathfrak{g}_{\mathrm{std}}}$.  
% Recall the expression of $G_M$ from Eq.~\eqref{eq:grover}.  
% We will establish the stronger claim that each $e_k(X_1, \dots, X_n)$ is a polynomial of $B$.
%belongs to $\mathcal{A}_{\mathfrak{g}_{\mathrm{std}}}$.
%
It is well known that the
elementary symmetric polynomials $e_0$, $e_1$, $\dots,e_n$ in variables $x_1,\dots,x_n$ can be expressed inductively in terms of the power sums
\[
p_l := \sum_{j=1}^n x_j^l
\]
according to Newton's formula
\[
e_k = \frac1k \sum_{l=1}^k (-1)^{l-1} p_l e_{k-l}.
\]
Note that, because $X_j^2=I$, we have:
\[
p_l(X_1, \dots, X_n) = \begin{cases}
    nI, & \text{$l$ even,} \\
    B, & \text{$l$ odd.} \\
\end{cases}
\]
It follows by induction that all $e_k(X_1, \dots, X_n)$ are polynomials of $B$.
%thus completing the proof.
\end{proof}
%
% We proceed by induction on $j$.  
% For the base case, note that $e_0 = I$ and the standard mixer is
% \[
%    B \;=\; X_1 + \cdots + X_n \;=\;  e_1(X_1,\dots,X_n).
% \]
% Thus, both $e_0$ and $e_1$ belong to $\mathcal{A}_{\mathfrak{g}_{\mathrm{std}}}$.  

% Now assume inductively that $e_t(X_1, \dots, X_n) \in \mathcal{A}_{\mathfrak{g}_{\mathrm{std}}}$ for all $0 \leq t < j$.  
% Each $e_j$ can be expressed as a linear combination of powers of $B$ and lower degree elementary symmetric polynomials:
% \[
%    e_j \;=\; B^j - \sum_{0 \leq t < j} \alpha_t \, e_t,
% \]
% for suitable coefficients $\alpha_t \in \mathbb{R}$.  
% By the inductive hypothesis, every $e_t$ with $t<j$ is in $\mathcal{A}_{\mathfrak{g}_{\mathrm{std}}}$, and clearly $B^j \in \mathcal{A}_{\mathfrak{g}_{\mathrm{std}}}$ as well.  
% Hence $e_j \in \mathcal{A}_{\mathfrak{g}_{\mathrm{std}}}$.  

% By induction, all elementary symmetric polynomials $e_j$ belong to $\mathcal{A}_{\mathfrak{g}_{\mathrm{std}}}$, and therefore so does $G_M$.  
% This proves the containment of algebras $\mathcal{A}_{\mathfrak{g}_\xi} \subseteq \mathcal{A}_{\mathfrak{g}_{\mathrm{std}}}$.

\subsection{Proof of Theorem \ref{LargestCommutantThm}}

% The commutant of the DLA arising from a QAOA algorithm consists of all operators in the ambient algebra of matrices acting on the Hilbert space, \( \mathrm{End}(W) \), that commute with every element of that Lie algebra. Since the DLA is generated by \( G_M \) and \( H_P \), it suffices to check commutation with these two operators. 

We begin by determining the structure of the commutant (also called centralizer) \( \mathcal{C}(H_P) \) of the problem Hamiltonian \( H_P \). 
%
% \begin{defn}\label{CentralizerDef}
% The \emph{centralizer} of \( H_P \) in \( \mathrm{End}(W) \) is defined as
% \[
% Z_{H_P} := \{ g \in \mathrm{End}(W) \mid g H_P = H_P g \},
% \]
% that is, the set of all linear operators on \( W \) that commute with \( H_P \).
% \end{defn}
%
The following characterization is immediate from the fact that \( H_P \) acts on each eigenspace \( W_{\lambda_j} \) (\( 1 \le j \le r \)) as $\lambda_j$ times the identity operator.
%as a scalar operator \( \lambda_j I_{n_j} \), where \( I_{n_j} \) denotes the identity matrix of size \( n_j \times n_j \).

\begin{lem}
The centralizer of $H_P$ in \( \mathrm{End}(W) \) is given by
\begin{equation*}
\mathcal{C}(H_P) 
= \bigoplus_{j=1}^{r} \operatorname{End}(W_{\lambda_j}) \pi_{W_{\lambda_j}},
\end{equation*}
where $\pi_{W_{\lambda_j}}$ denotes the orthogonal projection onto the eigenspace $W_{\lambda_j}$.
\label{HpCentralizer}
\end{lem}

In a basis that respects the decomposition of \( W \) from Eq.~\eqref{LevelSetDecomp}, the centralizer \( \mathcal{C}(H_P)  \) consists of block-diagonal matrices of the form
\[
\left(
  \begin{array}{c|c|c|c}
    M_1 & 0 & \cdots & 0 \\
    \hline
    0 & M_2 & \ddots & 0 \\
    \hline
    \vdots & \ddots & \ddots & \vdots \\
    \hline
    0 & 0 & \cdots & M_{r}
  \end{array}
\right),
\]
where each block \( M_j \) corresponds to a distinct eigenvalue \( \lambda_j \) of \( H_P \).  
The number of blocks, \( r \), equals the number of unique eigenvalues of \( H_P \), and the size of each block \( M_j \) is \( n_j \), the multiplicity of \( \lambda_j \).

% \begin{rmk}
%    The commutant of any Lie algebra, containing $H_P$ is a subalgebra of the centralizer \( Z_{H_P} \), described in Theorem~\ref{HpCentralizer}.  
% \end{rmk}

% Theorem \ref{LargestCommutantThm} provides an explicit upper bound
%  \[
%     1+\sum_{j=1}^{d} (n_j - 1)^2+\sum_{j=d+1}^{r} n_j^2.
% \]
% on the dimension of the commutant for DLAs generated by a given problem Hamiltonian together with a mixer whose negated Hamiltonian has the fixed initial unit vector \( \xi \) as its ground state. The proof is as follows.

\begin{proof}[Proof of Theorem \ref{LargestCommutantThm}]
Note that, by definition,
\[
\mathcal{C}(\mathfrak{g}_\xi) = \mathcal{C}(G_M,H_P) = \mathcal{C}(G_M) \cap \mathcal{C}(H_P).
\]
%For the first part, observe that since \( \ket{\xi}\) is the ground state of \( -H_M \), any 
%invariant, and consequently also preserve its orthogonal complement \( ( \mathbb{C}\ket\xi )^\perp \).
Any operator $T\in\mathcal{C}(H_P)$, commuting with \( H_P \), must preserve its eigenspaces $W_{\lambda_j}$ (\( 1 \le j \le r \)), as in Lemma \ref{HpCentralizer}. Similarly, any $T\in\mathcal{C}(G_M)$ must preserve the eigenspaces \( \mathbb{C}\ket\xi \) and \( ( \mathbb{C}\ket\xi )^\perp \) of \( G_M \).

Hence, any $T\in\mathcal{C}(\mathfrak{g}_\xi)$ preserves the intersections 
\[ 
\widetilde{W}_{\lambda_j} = W_{\lambda_j} \cap (\mathbb{C}\ket\xi)^\perp,
\qquad 1 \leq j \leq d. 
\]
Next, recall that 
\[
\ket{\xi}= \sum\limits_{j=1}^{d} c_j\ket{\xi_j}, \qquad \ket{\xi_j} \in W_{\lambda_j},
\]
is the unique expression of $\ket\xi$ as a sum of normalized eigenvectors of $H_P$. If $T\ket\xi = \alpha\ket\xi$ for some $\alpha\in\mathbb C$, it follows that
$T\ket{\xi_j} = \alpha\ket{\xi_j}$ for all \( 1 \le j \le d \).
This implies that $T$ acts as $\alpha$ times the identity on the subspace $W_0$ spanned by $\ket{\xi_1},\dots,\ket{\xi_d}$.

Therefore, every $T\in\mathcal{C}(\mathfrak{g}_\xi)$ lies in the right-hand side of Eq.~\eqref{eq:Cgxi}.
Conversely, any operator $T$ in the the right-hand side of Eq.~\eqref{eq:Cgxi} will satisfy the above properties, and hence will commute with both $G_M$ and $H_P$.
This proves the equality \eqref{eq:Cgxi} and concludes the proof of Theorem \ref{LargestCommutantThm}.
\end{proof}

Note that the above proof can be applied more generally to a different mixer $H_M$ in place of $G_M$, and gives a direct proof of Corollary \ref{cor:maxcomm}.

\subsection{Proofs of Theorems \ref{thm:ExpVar} and \ref{thm:BPtheorem}}

% The goal of this section is to compute the asymptotic behavior of the variance of the loss function for QAOA with a mixer Hamiltonian of type $G_{M} $. Let the parameterized quantum circuit for QAOA be defined by the unitary
% \[
% U(\boldsymbol{\beta}, \boldsymbol{\gamma}) = \prod_{i=1}^p e^{-i \beta_i H_M} e^{-i \gamma_i H_P},
% \]
% where \( H_M \) is the mixer Hamiltonian, \( H_P \) is the problem Hamiltonian, and the initial state is \( \ket{\xi}\).

% The output is obtained by measuring the observable \( O = H_P \), leading to the loss function \(\ell_{\boldsymbol{\beta}, \boldsymbol{\gamma}} (\rho, H_P) \), defined by the formula in \eqref{LossFnEqn}.

Due to Theorems~2 and~3 of~\cite{RBSKMLC},
for a deep enough quantum circuit \eqref{qaoa-chain} (i.e., for sufficiently large $p$), the expectation and variance of the loss function \eqref{LossFnEqn2} over the space of parameters $(\boldsymbol{\beta}, \boldsymbol{\gamma})$ can be approximated by the corresponding expectation and variance taken over the dynamical Lie group $G=e^{\mathfrak{g}}$. The latter have been found in Theorem~1 of \cite{RBSKMLC}, which we recall here.

Let us decompose the dynamical Lie algebra $\mathfrak{g}$ as
\[
  \mathfrak{g} = \mathfrak{g}_1 \oplus \cdots \oplus \mathfrak{g}_k \oplus \mathfrak{z},
\]
into simple compact Lie algebras \( \mathfrak{g}_j \) together with a (possibly trivial) center \( \mathfrak{z} \) 
(this can be done in general, by Proposition A.1 in \cite{WKKB1}).

Then the variance of the loss function with respect to the Haar measure on the dynamical Lie group $G=e^{\mathfrak{g}}$ is:
%Explicitly, one has
\begin{equation}
%\begin{aligned}
   \operatorname{Var}_{G}
  \bigl[\ell_{\boldsymbol{\beta}, \boldsymbol{\gamma}} (\rho, O) \bigr] 
  = \sum_{j=1}^{k} \frac{\mathcal{P}_{\mathfrak{g}_j}(\rho) \, \mathcal{P}_{\mathfrak{g}_j}(O)}{\dim(\mathfrak{g}_j)},
%\end{aligned}
\label{VarianceEqn}
\end{equation}
provided that $O\in i\mathfrak{g}$ (see \cite{RBSKMLC}).
In the above formula,
$\mathcal{P}_{\mathfrak{g}_j}(H) := \operatorname{Tr}\!\bigl[H_{\mathfrak{g}_j}^2\bigr]$
and \( H_\mathfrak{g_j} \) denotes the orthogonal projection of an operator \(H\) onto the subalgebra \( \mathfrak{g}_j \subseteq \mathrm{End}(W)\).  

For the expectation, we have (again by \cite{RBSKMLC}, Theorem~1):
\begin{equation}
\mathbb{E}_{G}\bigl[\ell_{\boldsymbol{\beta}, \boldsymbol{\gamma}} (\rho, O) \bigr]
= \operatorname{Tr}\bigl[\rho_{\mathfrak{z}}O_\mathfrak{z} \bigr].
\label{ExpEqn}
\end{equation}

\begin{proof}[Proof of Theorem \ref{thm:ExpVar}]
We apply formulas \eqref{VarianceEqn} and \eqref{ExpEqn} to the case of GM-QAOA. Then the  the dynamical Lie algebra $\mathfrak{g}_{\xi}$ admits the decomposition
\[
    \mathfrak{g}_{ \xi} \cong \mathfrak{su}(d) \oplus \mathfrak{z},
\]
where the center \( \mathfrak{z} \) is isomorphic to  \( \mathfrak{u}(1)^{\oplus2} \)  (see Theorem~\ref{DLAMainThm}).
 Since the contribution of the center to the variance vanishes, we get
 %only the simple part contributes %Hence, we have the simplified expression:
\[
\operatorname{Var}_{G_\xi}\bigl[\ell_{\boldsymbol{\beta}, \boldsymbol{\gamma}} (\rho, O) \bigr] 
= \frac{\mathcal{P}_{\mathfrak{su}(d)}(\rho) \, \mathcal{P}_{\mathfrak{su}(d)}(O)}{d^2-1}.
\]
We will compute this expression for the density matrix $\rho=\ket{\xi}\bra{\xi}$ corresponding to the initial state $\ket\xi$, and the observable $O=H_P$.

Recall from Eq.~\eqref{eq:restate} that the copy of $\mathfrak{su}(d)$ inside $\mathfrak{g}_{\xi}$ is realized as $\mathfrak{su}(W_{0})$.
Let \( \rho_{\mathfrak{u}(W_{0})} \) denote the restriction of $\rho$
%the rank-one operator \( \ket{\xi}\bra{\xi} \) 
to the Lie algebra \( \mathfrak{u}(W_{0}) \). Then
\[
\rho_{\mathfrak{u}(W_{0})} = \ket{v}\bra{v}, \quad \text{where } \ket{v} = (c_1, c_2, \dots, c_{d})^T,
\]
with \( \braket{v}{v} =\braket{\xi}{\xi} = 1 \). In the basis \( \{ \ket{\xi_1}, \dots, \ket{\xi_d} \} \), this coincides with  \( G_{M, 0} \), the restriction of \(G_{M} \) to the subspace \(W_{0}\) (see Eq.~\eqref{RestrictionsToTriv2}).

Since \( \rho_{\mathfrak{u}(W_{0})} \) is a rank-one projection, we have:
\[
\rho_{\mathfrak{u}(W_{0})}^2 = %v^T v v^T v = v^T v = 
\rho_{\mathfrak{u}(W_{0})}, \qquad
\operatorname{Tr}\bigl[\rho_{\mathfrak{u}(W_{0})}\bigr] = 1.
\]
To project further to the subalgebra of traceless matrices \( \mathfrak{su}(W_{0}) \subset \mathfrak{u}(W_{0}) \), we subtract the scalar part:
\[
\rho_{\mathfrak{su}(W_{0})} = \rho_{\mathfrak{u}(W_{0})} - \frac{1}{d} I_{d}.
\]
% Note that \( \operatorname{Tr}(\rho_{\mathfrak{u}(W_{0})}^2) = \operatorname{Tr}(\rho_{\mathfrak{u}(W_{0})}) = 1 \), since the operator is a rank-one projection.

We now compute the squared norm of the traceless component:
%\begin{equation}
\begin{align*}
\mathcal{P}&{}_{\mathfrak{su}(W_{0})}(\rho) = \operatorname{Tr} \!\left[ \left( \rho_{\mathfrak{u}(W_{0})} - \frac{1}{d} I_{d} \right)^2 \right] \\
&=\operatorname{Tr} \!\left[ \rho_{\mathfrak{u}(W_{0})}^2 - \frac{2}{d} \rho_{\mathfrak{u}(W_{0})} + \frac{1}{d^2} I_{d} \right] \\
&= \operatorname{Tr}[\rho_{\mathfrak{u}(W_{0})}^2] - \frac{2}{d} \operatorname{Tr}[\rho_{\mathfrak{u}(W_{0})}] + \frac{1}{d^2} \operatorname{Tr}[I_{d}] \\
&=1 - \frac{2}{d} + \frac{1}{d^2} d \\
&= 1 - \frac{1}{d}\,.
\end{align*}
%\end{equation}

% \begin{rmk}
%     If the initial state \(\ket{\xi}\) is not an eigenvector of  \(H_P\), then \( d\geq 2 \) and
%     \[
%         \mathcal{P}_{\mathfrak{su}(d)}(\rho) \geq \frac{1}{2}.
%     \]
% \end{rmk}

We now consider the projection of the observable \( H_P \) to the Lie algebra \( \mathfrak{u}(W_{0}) \). According to Eq.~\eqref{RestrictionsToTriv1}, this projection is given in the orthonormal basis \( \{ \ket{\xi_1}, \dots, \ket{\xi_d} \} \) by the diagonal matrix:
\[
    H_{P, 0} = \operatorname{diag} \left( \lambda_1, \lambda_2, \dots, \lambda_{d} \right).
\]
Let us define the following trace quantities:
\[
\begin{aligned}
    L_1 &:= \operatorname{Tr}[H_{P, 0}] = \sum\limits_{j=1}^{d} \lambda_j, \\
    L_2 &:= \operatorname{Tr}[H_{P, 0}^2] = \sum_{j=1}^{d} \lambda_j^2.
\end{aligned}
\]
Then the projection of \( H_P \) to the traceless Lie algebra \( \mathfrak{su}(W_{0}) \) is given by:
\[
H_{P, 0} - \frac{L_1}{d} I_{d}.
\]
From here, we find
\[
\begin{aligned}
    \mathcal{P}&{}_{\mathfrak{su}(W_{0})}(H_P) 
    = \operatorname{Tr}\!\left[\left( H_{P, 0} - \frac{L_1}{d} I_{d} \right)^2\right] 
    \\
    &=\operatorname{Tr}[H_{P, 0}^2] - \frac{2L_1}{d} \operatorname{Tr}[H_{P, 0}] + \frac{L_1^2}{d^2} \operatorname{Tr}[I_{d}] \\
    &= L_2 - \frac{2L_1^2}{d} + \frac{L_1^2}{d} \\
    &= L_2 - \frac{L_1^2}{d}.
\end{aligned}
\]
%Using that \( \operatorname{Tr}(I_{d}) = d \), we obtain:
% \begin{equation*}
%     \mathcal{P}_{\mathfrak{su}(W_{0})}(H_P) 
%     = L_2 - \frac{2L_1^2}{d} + \frac{L_1^2}{d} = L_2 - \frac{L_1^2}{d}.
% \end{equation*}

We further observe that:
\[
%\begin{aligned}
    L_1^2 
    %&= \left( \sum_{j=1}^{d} \lambda_j \right)^2 \\
    = \sum_{j=1}^{d} \lambda_j^2 + 2 \sum_{1 \leq i < j \leq d} \lambda_i \lambda_j,
    %&=L_2 + 2 \sum_{1 \leq i < j \leq d} \lambda_i \lambda_j,
%\end{aligned}
\]
and we can express the projection also as:
\[
\begin{aligned}
    \mathcal{P}&{}_{\mathfrak{su}(W_{0})}(H_P) \\
    %= \frac{d-1}{d} \sum_{j=1}^{d} \lambda_j^2 - \frac{2}{d} \sum_{1 \leq i < j \leq d} \lambda_i \lambda_j\\
    &=\sum_{j=1}^{d} \lambda_j^2 - \frac{1}{d} \left( \sum_{j=1}^{d}\lambda_j^2 + 2 \sum_{1 \leq i < j \leq d} \lambda_i \lambda_j \right)\\  
    %&=\sum_{j=1}^{d} \lambda_j^2 - \frac{1}{d} \left( \sum_{j=1}^{d} \lambda_j \right)^2 \\
    &= \frac{1}{d}\sum\limits_{1 \leq i < j \leq d} \left( \lambda_i-\lambda_j \right)^2.
\end{aligned}
\]
%
% \begin{rmk}
% Let \( \zeta_{\Lambda} \) be a random variable uniformly distributed on the set of values attained by objective function \( \left\{\lambda_1, \lambda_2, \dots, \lambda_{d}\right\}. \) Then the quantity
% \[
% \frac{1}{d} \sum_{1 \leq i < j \leq d} \left( \lambda_i - \lambda_j \right)^2
% \]
% is equal to \( d \cdot \operatorname{Var}(\zeta_{\Lambda}) \), where \( \operatorname{Var}(\zeta_{\Lambda}) \) denotes the variance of \( \zeta_{\Lambda} \).
% \end{rmk}
%
Applying the formula for $\operatorname{Var}(\zeta_{\Lambda})$ from Eq.~\eqref{eq:variance}, we obtain the expression for $\operatorname{Var}_{G_\xi}\bigl[\ell_{\boldsymbol{\beta}, \boldsymbol{\gamma}} (\rho, H_P) \bigr]$ stated in Eq.~\eqref{varianceFormula}.

Next, we compute the expectation of the loss function. 
% Since \( H_P \in i\mathfrak{g}_{ \xi} \) (it is one of the generators), it follows from Theorem~1 in~\cite{RBSKMLC} that 
% \[
% \mathbb{E}_{\boldsymbol{\beta}, \boldsymbol{\gamma}}\left[\ell_{\boldsymbol{\beta}, \boldsymbol{\gamma}} (\rho, H_P) \right]
% = \operatorname{Tr}\left( \rho_{\mathfrak{z}}(H_P)_\mathfrak{z} \right).
% \]
Recall that the center \( \mathfrak{z} \) of \( \mathfrak{g}_{\xi} \) is $2$-dimensional, as described in Corollary~\ref{cor:center}. In the basis specified there, the orthogonal projections of \( \rho \) and \( H_P \) onto \( \mathfrak{z} \) are given by the following \( 2 \times 2 \) matrices:

\[
\rho_\mathfrak{z} = 
\begin{pmatrix}
0 & 0 \\
0 & \frac{1}{d}
\end{pmatrix}, 
\qquad
(H_P)_\mathfrak{z} = 
\begin{pmatrix}
1 & 0 \\
0 & \frac{L_1}{d}
\end{pmatrix}.
\]
%where \(L=\frac{1}{d} \sum\limits_{j=1}^{d} \lambda_j\). 
The trace of their product  is then
\[
\operatorname{Tr}\bigl[
\rho_{\mathfrak{z} }(H_P)_\mathfrak{z}
\bigr]= \frac{1}{d^2} \sum\limits_{j=1}^{d} \lambda_j.
\]
This proves Eq.~\eqref{ExpectationOfLossFunction} and completes the proof of Theorem \ref{thm:ExpVar}.
% Consequently, we arrive with the formula for the expected value of $\ell_{\boldsymbol{\beta}, \boldsymbol{\gamma}} (\rho, H_P)$ from Eq.~\eqref{ExpectationOfLossFunction}. 
% This completes the proof of Theorem \ref{thm:ExpVar}.
\end{proof}

Finally, we present the proof of Theorem \ref{thm:BPtheorem}.

\begin{proof}[Proof of Theorem \ref{thm:BPtheorem}]
We first note that for any collection of distinct integers \(\lambda_1, \dots, \lambda_d\), the variance \(\operatorname{Var}(\zeta_{\Lambda})\), given by Eq.~\eqref{eq:variance}, admits the following lower bound:  
\[
\operatorname{Var}(\zeta_{\Lambda}) \ge \frac{(d-1)(d+1)}{12} \,.
\]
%(See Eq.~\eqref{eq:variance} for the explicit expression for \(\operatorname{Var}(\zeta_{\Lambda})\).)  
Indeed, if we impose an ordering \(\lambda_1 > \cdots > \lambda_{d}\), then $\lambda_i - \lambda_j \ge j-i$ for $i<j$, and
we have
\begin{align*}
\sum_{1 \le i < j \le d} (\lambda_i - \lambda_j)^2 
&\ge \sum_{1 \le i < j \le d} (i - j)^2 \\
&= \frac{d^2 (d-1)(d+1)}{12},
\end{align*}
where equality is achieved when the integers \(\lambda_i\) are consecuitive.
%form an arithmetic progression with unit spacing.

Applying formula~\eqref{varianceFormula}, we obtain
\[
\operatorname{Var}_{\boldsymbol{\beta}, \boldsymbol{\gamma}}\bigl[\ell_{\boldsymbol{\beta}, \boldsymbol{\gamma}} (\rho, H_P)\bigr] \ge \frac{d-1}{12} \,,
\]
and consequently,
\[
\operatorname{Var}_{\boldsymbol{\beta}, \boldsymbol{\gamma}}\bigl[\ell_{\boldsymbol{\beta}, \boldsymbol{\gamma}} (\rho, \widehat{H}_P)\bigr] \ge \frac{d-1}{12 \, \lambda_{\max}^2} \,.
\]
Since \(\ket{\xi}\) is not an eigenvector of \(H_P\), we have $d \ge 2$. 

It remains to bound $|\lambda_{\max}|$ from above. 
As each term \( f_j \) contributes at most \( M \) to the total value of $F(x)$ for $x\in\mathbb{B}^n$,
we get
%Consequently, the maximum possible value of \( F(x) \) obeys the bound
\[
|\lambda_{\max}| \le MT \le M \binom{n}{s} \le \frac{M n^s}{s!} \,.
\]
This proves Eq.~\eqref{VarLowerBound} and completes the proof of Theorem \ref{thm:BPtheorem}.
% Therefore, the lower bound 
% \[
% \frac{(s!)^2}{12 M^2 n^{2s}}
% \]
% from Eq.~\eqref{VarLowerBound} holds.
\end{proof}

% \bibliographystyle{plain}
% \bibliography{referencesMain}

\end{document}